\documentclass[sigconf, nonacm]{acmart}

\usepackage{cleveref}
\usepackage{pifont}
\usepackage{bbm}
\usepackage{algorithm}
\usepackage[noend]{algpseudocode}
\usepackage{subcaption}

\newcommand{\Credit}{\mathcal{C}}
\newcommand{\Bad}{\mathcal{B}}
\newcommand{\Good}{\mathcal{G}}
\newcommand{\gammadeg}{d_\gamma}
\newcommand{\Exp}[1]{\mathbb{E}_{#1}}
\newcommand{\Prob}[1]{\mathbb{P}_{#1}}
\newcommand{\ie}{i.e.,}
\newcommand{\cs}{\tau}

\newtheorem{note}{Note}

\algnewcommand{\IIf}[1]{\State\algorithmicif\ #1\ \algorithmicthen}
\algnewcommand{\EndIIf}{\unskip}

\newcommand\vldbdoi{XX.XX/XXX.XX}

\newcommand\vldbavailabilityurl{https://github.com/SimonePellegrini/Dynamic-Quasi-Clique}

\begin{document}
\title{Detecting Large Quasi-cliques on Dynamic Networks}

\author{Luciano Gual\`{a}}
\orcid{0000-0001-6976-5579}
\affiliation{%
  \institution{\emph{Tor Vergata} University of Rome}
  \city{Rome}
  \state{Italy}
}
\email{guala@mat.uniroma2.it}

\author{Simone Pellegrini}
\affiliation{%
  \institution{\emph{Tor Vergata} University of Rome}
  \city{Rome}
  \state{Italy}
}
\email{simonepell2003@gmail.com}

\author{Luca Pep\`{e} Sciarria}
\orcid{0000-0003-4432-6099}
\affiliation{%
  \institution{\emph{Tor Vergata} University of Rome}
  \city{Rome}
  \state{Italy}
}
\email{luca.pepesciarria@gmail.com}

\author{Alessandro Straziota}
\orcid{0009-0008-4543-786X}
\affiliation{%
  \institution{\emph{Tor Vergata} University of Rome}
  \city{Rome}
  \state{Italy}
}
\email{alessandro.straziota@uniroma2.it}

\begin{abstract}
Motivated by the problem of detecting large and cohesive groups of vertices in real networks, the task of finding large \emph{quasi-cliques} has attracted considerable attention across different research areas. From a computational complexity perspective, strong inapproximability results are known for this problem, yet several heuristics have been proposed to identify large quasi-cliques in real-world networks. Recently, [Pang \emph{et al.}, (WWW 2024)] introduced a similarity-based approach that represents the current state of the art.
In this work, we extend that approach to \emph{dynamic} networks, thereby addressing an open problem posed by [Pang \emph{et al.}, (WWW 2024)]. 

We first present a Baseline fully dynamic algorithm where edges of the network can be both inserted and deleted. The algorithm exactly maintains the same quasi-clique returned by the algorithm by Pang et al. on the current graph, with update time $\widetilde{O}(\Delta)$, where $\Delta$ is the maximum degree. We then focus on the practically relevant incremental case, where only edge insertions are allowed, and design an algorithm with $O(\log \Delta)$ update time. This method leverages a novel technique for dynamically maintaining accurate estimates of vertex $\gamma$-degrees, a core component of framework by Pang et al., and achieves up to $207\times$ speed-up over the Baseline while preserving comparable solution quality. Finally, we extend the approach to the fully dynamic setting, supporting both insertions and deletions, obtaining up to $21\times$ speed-up with limited and acceptable loss in quasi-clique size and density. We provide a formal analysis of our algorithms and validate them through an extensive set of experiments on real-world datasets.
\end{abstract}

\maketitle


\ifdefempty{\vldbavailabilityurl}{}{
\vspace{.3cm}
\begingroup\small\noindent\raggedright
\textbf{Artifact Availability:}\\
The source code, data, and/or other artifacts have been made available at \url{\vldbavailabilityurl}.
\endgroup
}

\section{Introduction}

Identifying large dense subgraphs is a fundamental task in network analysis, with applications spanning community detection in online social platforms \cite{GirvanNewman2002,LeskovecLangMahoney2010}, functional module discovery in biological networks \cite{BaderHogue2003,SpirinMirny2003}, fraud detection in financial systems \cite{HooiEtAl2016,PanditEtAl2007}, and product recommendation in e-commerce \cite{LindenSmithYork2003,McAuleyLeskovec2015}. In many of these domains, the presence of a large and cohesive group of vertices often signals a meaningful structural or functional pattern. As modern networks routinely comprise millions or even billions of edges, designing algorithms capable of efficiently detecting large dense substructures has become a central challenge in graph mining.

A widely adopted notion to capture dense connectivity is that of an $\alpha$-quasi-clique. Given a graph $G=(V,E)$ and a parameter $\alpha \in (0,1]$, a subset of vertices $S \subseteq V$ is said to be an $\alpha$-quasi-clique if $S$ induces a subgraph of $G$ having at least a fraction $\alpha$ of the edges of a complete graph over the same vertices (see \Cref{fig:quasi-clique-example}). When $\alpha = 1$, this definition coincides with the classical notion of a clique. For $\alpha < 1$, quasi-cliques relax the all-to-all connectivity requirement, making them significantly more robust to noise and missing edges -- features that are ubiquitous in real-world data. As a result, $\alpha$-quasi-cliques offer a principled and flexible alternative to cliques, balancing structural coherence with tolerance to imperfections in the observed network.

\begin{figure}
    \centering
    \includegraphics[width=0.7\linewidth]{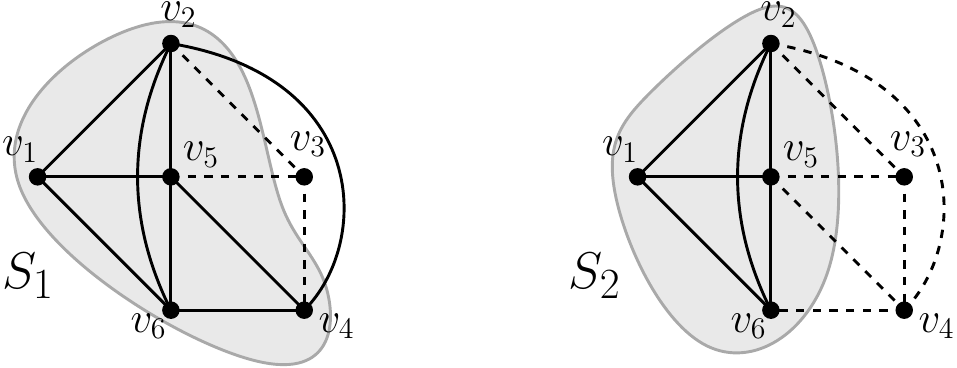}
    \caption{%
    An example of two $0.9$-quasi-clique on the same graph.
    The set $S_2= \{v_1,v_2,v_5,v_6\}$ (on the right) has edge density $1$, since it is a clique, while the set $S_1 = \{v_1,v_2,v_4,v_5,v_6\}$ (on the left) has a lower edge density of $0.9$ but larger size $5$.%
    }
    \label{fig:quasi-clique-example}
    \Description{An example of two $0.9$-quasi-clique on the same graph.}
\end{figure}

From a computational standpoint, however, finding a large quasi-clique is notoriously difficult. The problem is NP-hard \cite{PattilloVBB13} for any fixed $\alpha$ and, in fact, inherits strong inapproximability results: since the maximum clique problem is a special case obtained by setting $\alpha=1$, computing a maximum $\alpha$-quasi-clique cannot be approximated within a factor of $n^{1-\varepsilon}$ for any $\varepsilon > 0$, unless P=NP \cite{Zuckerman07}. This inherent hardness has motivated a substantial body of work devoted to the design of effective heuristics and approximation strategies \cite{AbelloRS02,ChenCPWLZY21,DjeddiHB19,KonarS20,MitzenmacherPPT15,PintoRRR21,PintoRRP18,TsourakakisBGGT13}. A recent milestone in this line of research is the work of Pang et al. \cite{pang}, which introduced the \texttt{NBSim} and \texttt{FastNBSim} algorithms. Their key insight is that vertices within a quasi-clique tend to exhibit similar neighborhoods; by exploiting neighborhood similarity -- further accelerated via \emph{min-hash} approximations \cite{broder1997resemblance} -- \texttt{FastNBSim} achieves orders-of-magnitude speedups over prior methods while maintaining high solution quality on massive real-world graphs.

Despite these advances, existing approaches -- including \texttt{NBSim} and \texttt{FastNBSim} -- operate in the static setting, where the entire graph is given as input and the quasi-clique is computed from scratch. Yet many real-world networks are inherently dynamic: edges are continuously inserted and deleted as users interact, transactions occur, or relationships evolve. Notably, Pang et al. explicitly mention the extension of their techniques to dynamic graphs as an open problem. Addressing this challenge is precisely the focus of our work. We study the problem of maintaining a large quasi-clique in a network subject to an arbitrary sequence of edge insertions and deletions, aiming to update the maintained solution significantly faster than recomputing it from scratch after each modification.

\subsection{Our results}

To the best of our knowledge, we provide the first systematic study of the dynamic maintenance of large quasi-cliques under edge updates, combining theoretical guarantees with an extensive experimental evaluation.

\paragraph{A fully dynamic baseline.}
Using relatively standard dynamic graph techniques, we design a first fully dynamic algorithm that maintains a large quasi-clique under both edge insertions and deletions. Importantly, at any point in time, the maintained quasi-clique coincides exactly with the one that would be computed by the state-of-the-art \texttt{NBSim} algorithm of Pang et al. on the current graph. The algorithm supports updates in time $\widetilde{O}(\Delta)$\footnote{The notation $\widetilde{O}$ hides poly-logarithmic factors.}, where $\Delta$ is the maximum degree of the current graph. Although conceptually simple, this solution already yields a substantial improvement over recomputing the quasi-clique from scratch at query time. For this reason, we consider it a natural and robust Baseline against which more refined heuristics can be compared.

\paragraph{An incremental algorithm with logarithmic update time.} 
We then focus on the incremental setting, where only edge insertions are allowed. This case is particularly relevant in practice: many real networks -- such as collaboration graphs, citation networks, and certain interaction or communication networks -- evolve primarily through the addition of new edges, while deletions are rare, expensive to process, or semantically less meaningful. Moreover, incremental models naturally capture scenarios in which data is accumulated over time (e.g., transaction logs or streaming interactions), making efficient insertion-only maintenance highly desirable.

For this setting, we design an algorithm with amortized\footnote{The amortized analysis is a well-known method originally introduced in \cite{tarjanAmortizedComplexity} that allows to compute tight bounds on the cost of a sequence of operations, rather than the worst-case cost of an individual operation. In more detail, an operation has amortized cost $\overline{c}$ if any arbitrary sequence of $\eta$ operations has a overall cost of at most $\overline{c} \eta$.} update time $O(\log \Delta)$.
In \Cref{sec:incremental} we analyze the effectiveness of the incremental algorithm under the well-known \emph{random permutation model} (that is, where the sequence of insertions is sampled uniformly at random from all possible permutations of edges).
The random permutation model has long served as a useful framework for understanding the behavior of dynamic graph algorithms,
yielding theoretical predictions that have often been confirmed in practice across a broad range of dynamic scenarios \cite{BecchettiCGSSS25,counting_triangles,robust_lowerbound_stream,HanauerHS22,kapralovSODA14,stream_alg_survey,MonemizadehMPS17,PengSODA18}.

Nevertheless, it is natural to question whether such a model accurately reflects the scenarios addressed in this work. To investigate this aspect, we present an extensive experimental evaluation in \Cref{sec:experiments} on incremental graphs resulting from real insertion sequences, suggesting that, for our algorithm, the random permutation model provides a sufficiently accurate description of real scenarios.

From an efficiency perspective, the incremental algorithm achieves substantial speedups, up to $207\times$ on large graphs, compared to the Baseline. Technically, the algorithm relies on a novel technique for dynamically maintaining accurate estimates of the $\gamma$-degree of each vertex, a key ingredient in the approach of Pang et al. \cite{pang}. This mechanism enables us to localize and sharply reduce the work performed at each insertion.

\paragraph{Extension to the fully dynamic setting.}
Finally, we suitably extend the incremental approach to handle deletions, obtaining a fully dynamic algorithm capable of supporting arbitrary edge updates. As expected, accommodating deletions -- an inherently more challenging scenario -- incurs additional overhead, and the resulting algorithm is slower than the incremental variant. Nevertheless, although for this algorithm is not possible to prove  an $o(\Delta)$ update time in the worst case, on several real-world datasets it still achieves significant speedups, up to $21\times$ over the Baseline, while incurring only a limited and acceptable loss in solution quality in terms of both size and density of the maintained quasi-clique.

All our algorithms are validated through an extensive experimental study on multiple dynamic and static real-world networks. Overall, our results demonstrate that dynamic maintenance of large quasi-cliques is not only feasible but can be performed efficiently in practice, offering different quality-efficiency trade-offs tailored to the application scenario.

\subsection{Related works}
As we already pointed out, we are not aware of previous work on the dynamic maintenance of large quasi-cliques, although dynamic algorithms has a long tradition in the theoretical computer science community (see \cite{HanauerHS22} for a recent survey) as well as in the graph mining community (e.g., \cite{BecchettiCGSSS25,SunCS20,ZhangGGW26}). Among existing research directions, the line of work most closely related to ours is that on dynamic algorithms for structural clustering \cite{RuanGWW21,ZhangGGW26,ZhaoGRB0025}. An attentive reader may observe similarities between the approach of Pang et al. and the well-known SCAN algorithm for structural clustering. In particular, both approaches exploit forms of neighborhood similarity to identify cohesive vertex sets. However, to the best of our knowledge, it remains unclear whether the techniques developed in \cite{RuanGWW21,ZhangGGW26,ZhaoGRB0025} to make (variants of) SCAN dynamic can be adapted to the problem of maintaining large quasi-cliques.

In the static setting, beyond the aforementioned work of Pang et al. \cite{pang} and other contributions focusing on fast heuristics for the maximum quasi-clique problem \cite{AbelloRS02,ChenCPWLZY21,DjeddiHB19,KonarS20,MitzenmacherPPT15,PintoRRR21,TsourakakisBGGT13}, several papers propose exact algorithms for computing maximum quasi-cliques \cite{MarinelliPR21,MiaoB20,PattilloVBB13,RibeiroR19,VeremyevPBP16}. These exact approaches typically rely on sophisticated pruning and branch-and-bound strategies, but their computational cost limits their applicability to relatively small or medium-sized graphs. 

Another notion of dense structure (conceptually different from quasi-cliques) that has attracted significant attention is the \emph{densest subgraph problem}  \cite{Goldberg1984FindingAM,MitzenmacherPPT15,Tsourakakis15a,XuMFB24}. 
The problem has also been extensively studied in dynamic graphs, where efficient algorithms are known for maintaining approximate solutions under edge updates \cite{BeraBCG22,BhattacharyaHNT15,EpastoLS15,SawlaniW20}. Finally, quasi-cliques have been considered also in the context of \emph{temporal} graphs \cite{temporal_quasi_clique}.

\subsection{Paper organization}
The remainder of this paper is organized as follows.
In \Cref{sec:preliminaries}, we describe the \texttt{FastNBSim} algorithm by Pang \emph{et al.}~\cite{pang} and establish the necessary notation.
In \Cref{sec:baseline}, we demonstrate how \texttt{FastNBSim} can be extended via simple arguments to support edge insertions and deletions in $\widetilde{O}(\Delta)$ time.
In \Cref{sec:incremental}, we propose a faster incremental algorithm based on the ideas of \texttt{FastNBSim}, providing a thorough analysis of both its effectiveness and efficiency.
In \Cref{sec:fully}, we introduce a heuristic that extends this incremental approach to also support edge deletions, resulting in a fully dynamic algorithm.
Finally, \Cref{sec:experiments} presents an extensive experimental evaluation on both synthetic and real-world sequences to validate the performance of our algorithms.

\section{Preliminaries}\label{sec:preliminaries}
We consider unweighted and undirected graph $G=(V,E)$, where $V$ is the set of  vertices and $E \subseteq V \times V$ the set of edges.
We denote $n = \vert V \vert$ and $m = \vert E \vert$.
We say that two vertices $u,v$ are \emph{neighbors} in $G$ if $(u,v) \in E$.
The degree $d(u)$ of a vertex $u$ is the number of its neighbors, while the neighborhood $N(u)$ of $u$ is the set of vertices composed by $u$ and all its neighbors.
We denote with $\Delta$ the maximum degree in $G$.

The definition of quasi-clique is based on that of \emph{edge-density}, that is, the fraction of edges that a subset of vertices has with respect to the number of possible edges.
Formally, let $S \subseteq V$ be any subset of vertices, and $E(S) \subseteq E$ the set of graph edges among all pairs of vertices in $S$.
The edge-density of $S$ is defined as 
\[\delta(S) = \frac{\vert E(S) \vert}{\binom{\vert S \vert}{2}} =  \frac{2 \vert E(S) \vert}{\vert S \vert (\vert S \vert - 1)}.
\]
For a given value $\alpha \in [0,1]$, we say that $S$ is an \emph{$\alpha$-quasi-clique} if $\delta(S) \geq \alpha$ (see \Cref{fig:quasi-clique-example}).

Other important definitions that will be used extensively throughout the paper are those of $\gamma$-neighbor and $\gamma$-degree. Formally, given a value $\gamma \in [0,1]$, a vertex $v$ is a $\gamma$-neighbor of $u$ if $v \in N(u)$ and $\vert N(v) \vert \geq \gamma \vert N(u) \vert$.
Hence, the $\gamma$-degree $\gammadeg(u)$ of a vertex $u$ is the number of its $\gamma$-neighbors, \ie\
\[
\gammadeg(u) = \vert \{ v \in N(u) \mid \vert  N(v) \vert \geq \gamma \vert N(u) \vert \} \vert.
\]
In Pang et al.~\cite{pang}, they propose a greedy method for extracting dense components based on \emph{similarity} between vertex neighborhoods.
As similarity measure, they use the \emph{containment score}, which indicates how well the neighborhood $N(u)$ of a vertex $u$ is contained within that of another vertex $v$.
Formally, the containment score of $u$ in $v$ is defined as 
$
\cs(u,v) = \frac{\vert N(u) \cap N(v) \vert}{\vert N(u) \vert}.
$
Observe that the containment score is non-symmetric (\ie\ $\cs(u,v)$ might be different from $\cs(v, u)$). 

Computing the exact containment score requires $\Omega(\Delta)$ time.
To efficiently estimate $\cs$ we can use the \emph{min-hash sketch} \cite{broder1997resemblance}.
Consider any permutation $\pi: [n] \to [n]$, the \emph{min-hash} $h_\pi(A)$ of a subset $A \subseteq [n]$ is defined as the minimum value of the permutation $\pi(A)$, \ie\ $h_\pi(A) = \min_{a \in A} \pi(a)$.
Given $k \geq 1$ permutations $\pi_1, \dots, \pi_k$ of $[n]$, we will refer to $\sigma(A) = (h_{\pi_1}(A), \dots, h_{\pi_k}(A))$ as the \emph{$k$-min-hash signature} of $A$, for every $A \subseteq [n]$.

The main property of the min-hash is that, if we sample uniformly at random a permutation $\pi$ from all the possible permutation of $[n]$, then for every pair of sets $A,B \subseteq [n]$ we have that $\Prob{\pi}\{h_\pi(A) = h_\pi(B)\} = J(A,B)$, where $J(A,B) = \frac{\vert A \cap B \vert}{\vert A \cup B \vert}$ is the \emph{Jaccard similarity} between $A$ and $B$.
Given the $k$-min-hash signatures of $N(u)$ and $N(v)$, it is well-known how to compute an unbiased estimation of $\cs(u,v)$ in $O(k)$ time (see for example \cite{pang}).

\subsection{\texttt{FastNBSim}: a fast algorithm for static graphs}\label{sec:pang-alg}

In this section we briefly describe the algorithm by Pang \emph{et al.} for extracting large quasi-cliques from a static graph \cite{pang}. 
Briefly, the algorithm first computes a set $S_u$ for each vertex $u$, which represents a potential solution, and then returns the largest one.
The procedure that computes $S_u$ is called \texttt{extract}, and is defined as follows:
Let $b \in (0,1]$ and $\gamma \in (0,1]$ be two input parameters.
We initially define $S_u$ as the set of vertices $v \in N(u)$ for which the containment score $\cs(u,v)$ of $u$ is at least $\gamma$.\footnote{Notice that, since $u \in N(u)$ we always have that $u \in S_u$.}
Then the parameter $b$ is used to guarantee a lower bound on the density of $S_u$.
In particular, if $\frac{\vert S_u \vert - 1}{\vert N(u) \vert} < b$ then $S_u$ is discarded (\ie\ $S_u$ is set to empty set).
In fact, by choosing appropriately the parameters $(b, \gamma)$, it is possible to show that if $S_u \neq \emptyset$, then its density is at least $\alpha$.

Since performing the \texttt{extract} procedure from each vertex would result in an inefficient algorithm in practice, an additional heuristic based on the $\gamma$-degree is exploited.
In fact, the $\gamma$-degree $\gammadeg(u)$ of a vertex $u$ is an upper-bound to the size of its quasi-clique $S_u$. Therefore, by scanning the vertices in non-increasing order of $\gammadeg$, it is sufficient to stop the scan when $\gammadeg(u)$ is smaller than the current solution found so far, thus \emph{pruning} many vertices and resulting in a fast algorithm in practice.
This algorithm is called \texttt{NBSim}.

Another bottleneck of the algorithm is the calculation of containment scores.
The authors provide a faster, but less accurate, version of \texttt{NBSim}, called \texttt{FastNBSim}, that uses $k$-min-hash signatures to estimate containment scores.
This algorithm results in significant speed-ups, with the small drawback of having slightly less dense quasi-clique.

Note that \texttt{NBSim} and \texttt{FastNBSim} differ only in that the former uses the exact containment scores $\cs(\cdot,\cdot)$ in the \texttt{extract} procedure, while the latter uses the estimator $\hat{\cs}(\cdot,\cdot)$. For this reason, we only report \texttt{FastNBSim} in \Cref{alg:FastNBSim,alg:extract}.

\begin{algorithm}
    \caption{\texttt{FastNBSim}}
    \label{alg:FastNBSim}
    \begin{flushleft}
    \textbf{Input:} a graph $G$, and $\gamma,b \in (0,1]$.
    \end{flushleft}
    \begin{algorithmic}[1]
        \State $S^* \gets \emptyset$\;
        \ForAll{$u$ in non-increasing $\gamma$-degree oreder}
            \IIf{$\gammadeg(u) < \vert S^* \vert$}
                \textbf{break}
            \EndIIf
            \State $S \gets \texttt{extract}(u, \gamma, b)$\;
            \IIf{$\vert S \vert > \vert S^* \vert$}
                $S^* \gets S$
            \EndIIf
        \EndFor
        \State \Return $S^*$
    \end{algorithmic}
\end{algorithm}

\begin{algorithm}
    \caption{\texttt{extract}}
    \label{alg:extract}
    \begin{flushleft}
    \textbf{Input:} a vertex $u$, and two parameters $\gamma,b \in (0,1]$.
    \end{flushleft}
    \begin{algorithmic}[1]
        \State $S \gets \emptyset$
        \ForAll{$v \in N(u)$}
            \State $\hat{\cs}(u,v) \gets$ estimate containment score of $u$ in     $v$\;
            \label{alg:extract:line:condition}
            \IIf{$\hat{\cs}(u,v) \geq \gamma$}
                $S \gets S \cup \{v\}$\;
            \EndIIf
    \EndFor
    \IIf{$\frac{|S| - 1}{|N(u)|} < b$} \label{line:clique_bound}
        $S \gets \emptyset$
    \EndIIf
    \State \Return $S$\;
    \end{algorithmic}
\end{algorithm}

\section{A fully-dynamic algorithm}\label{sec:baseline}

As a warm-up, in this section we describe a simple fully-dynamic algorithm with update time $O(\Delta \log n)$ and query time $O(1)$.
The algorithm, at any time, returns a quasi-clique $S^*$ that is \emph{exactly the same solution} that the \texttt{NBSim} algorithm would return if it were run on the graph at that instant.

The general idea is to keep all $S_u$ updated dynamically, along with other satellite information such as all the containment scores.
The techniques for keeping all information dynamically updated are quite simple yet effective, and have been exploited several times in other contexts \cite{BecchettiCGSSS25,RuanGWW21,ZhaoGRB0025}.
We see this algorithm as the baseline against which we will compare our faster solution.

Precisely, at each time, the algorithm keeps the following information updated for each vertex $u$:
\begin{itemize}
    \item the degree $d(u)$ and the neighborhood $N(u)$;

    \item the containment score $\cs(u,v)$, for every neighbor $v \in N(u)$;

    \item the quasi-clique $S_u$, that is, the set of vertices that would be returned by the \texttt{extract} procedure in \texttt{NBSim} (\ie\ using exact containment scores at Line~4 of \Cref{alg:extract:line:condition});

    \item the set $\widetilde{S}_u = \{v \in N(u) \mid \cs(u,v) \geq \gamma \}$. Notice that when $\frac{\vert S_u \vert - 1}{\vert N(u) \vert} \ge b$, then $S_u$ coincides with $\widetilde{S}_u$.
\end{itemize}

The algorithm keeps the $S_u$s sorted by size in a priority queue, so that the pointer to the largest can be returned in $O(1)$ time\footnote{observe that if we want to list all the elements of $S^*$ it takes $O(\vert S^* \vert)$.}.

We now describe how all the above information can be updated in $O(\Delta \log{n})$ time when a new edge $(x,y)$ is added to or removed from the graph.
Clearly, updating the degrees and neighborhoods of $x$ and $y$ can be done in constant time.
The main challenge is to efficiently update the affected containment scores.
Notice that only $O(d(x)+d(y)) = O(\Delta)$ containment scores may change, namely the scores $\cs(x, w), \cs(w,x)$ for $w \in N(x)$, and $\cs(y, w), \cs(w,y)$ for $w \in N(y)$. 
First of all, the scores $\cs(x,y)$ and $\cs(y,x)$ can be computed in $O(\Delta)$ time by looking at $N(u)$ and $N(v)$ explicitly.
Instead, all the others $O(\Delta)$ containment scores can be updated in constant time each using a simple closed formula, as shown by the following lemmata
(whose proof are deferred to Appendix~\ref{apx:deferred_proofs}).

Denote by $\cs'(\cdot,\cdot)$ the containment score \emph{after} the update (\ie\ the insertion or the deletion) of the edge $(x,y)$. 

\begin{lemma}\label{lm:cs_update_delete}
    Consider the deletion of edge $(x,y)$. For each $w \in N(x) \cap N(y) \setminus \{y\}$, we have $\cs'(w,x)=\frac{\cs(w,x)\vert N(w) \vert -1}{\vert N(w) \vert}$, and $\cs'(x,w)=\frac{\cs(x,w)\vert N(x)\vert -1}{\vert N(x)\vert -1}$.
    While, for each $w \in N(x) \setminus N(y)$, we have $\cs'(w,x)=\cs(w,x)$, and $\cs'(x,w)=\frac{\cs(x,w)\vert N(x)\vert}{\vert N(x)\vert -1}$.
\end{lemma}

\begin{lemma}\label{lm:cs_update_insert}
    Consider the insertion of edge $(x,y)$. For each $w \in N(x) \cap N(y) \setminus \{y\}$, we have $\cs'(w,x)=\frac{\cs(w,x)\vert N(w) \vert + 1}{\vert N(w) \vert}$, and  $\cs'(x,w)=\frac{\cs(x,w)\vert N(x)\vert +1}{\vert N(x)\vert +1}$. While, for each $w \in N(x) \setminus N(y)$, we have $\cs'(w,x)=\cs(w,x)$, and $\cs'(x,w)=\frac{\cs(x,w)\vert N(x)\vert}{\vert N(x)\vert +1}$.
\end{lemma}

After each containment score is updated, we can decide whether a neighbor $w \in N(x)$ belongs to $\widetilde{S}_x$ or not simply by checking the condition $\cs(x,w) \geq \gamma$. 
At the same time, however, all neighbors $w \in N(x) \setminus \{x\}$ only have to check whether to add or remove $x$ from $\widetilde{S}_w$.
The same holds for $y$.
We therefore update $O(\Delta)$ sets $\widetilde{S}$ in time $O(\Delta)$. 
Finally, if the condition $\frac{\vert \widetilde{S}_z \vert - 1}{\vert N(z) \vert} \geq b$ holds we set $S_z = \widetilde{S}_z$, otherwise $S_z = \emptyset$, and update the priority queue, for every $z \in N(x) \cup N(y)$.
The total worst-case update time is $O(\Delta \log{n})$.

\section{A faster solution for incremental networks}\label{sec:incremental}
As seen in \Cref{sec:pang-alg}, a key ingredient for the \texttt{FastNBSim} algorithm is the $\gamma$-degree, which (i) defines an order in which we explore the vertices for the quasi-clique extraction, and (ii) offers a pruning policy that allows us to not explore the entire graph.
The $\gamma$-degree could therefore be used for a dynamic algorithm, however, there is the problem of maintaining it dynamically.
In fact, even in the case of \emph{incremental graphs} (where only new edges are added), updating it would require a linear cost in the degree.

In this section we show how to efficiently maintain a $\gamma$-degree estimator for incremental graphs, where only new edges arrive and none are removed.
We then propose a dynamic algorithm which leverages the $\gamma$-degree estimator to efficiently updates and extracts large quasi-cliques.

\subsection{A \texorpdfstring{$\gamma$}{gamma}-degree estimator}
The idea for maintaining an estimate of the $\gamma$-degree is the following: for each vertex $u$ we maintain a counter $\Credit(u)$, which represents an amount of \emph{credits} that $u$ holds.
Now assume that at a time $t$ a new edge $(u,v)$ is inserted.
If, after the insertion of the edge $(u,v)$, the degree of $v$ is at least $\gamma$ times the degree of $u$, then we say that $v$ \emph{gives} a credit to $u$, \ie\ we just increment $\Credit(u)$ by one.
The intuition is that, since $\gamma \leq 1$, if $v$ gives a credit to $u$ this implies that $v$ was a $\gamma$-neighbor of $u$ at the time $t$ the edge $(u,v)$ was inserted\footnote{Observe that $u$ is always a $\gamma$-neighbor of itself hence we initialize $\Credit(u) = 1$.}.

We now show that the number of credits can be used to compute a good estimation of the $\gamma$-degree.
In the following analysis we assume the well-known \emph{random permutation model}, in which the edge insertion sequence is chosen uniformly at random among all the permutation of the edges set $E$ at the end of the sequence.

Let $d^t(u)$ be the degree of a vertex $u$ after the insertion of $t$ edges, and $d(u) = d^m(u)$ be the degree of $u$ at the end of the edge insertion sequence.
We now show that $d^t(u)$ concentrates around $\frac{t}{m}d(u)$.

\begin{lemma}\label{lemma:concentration_degree}
    For every $\varepsilon > 0$, for every time $t$, and for every vertex $u \in V$, we have
    \[
    \Prob{}\left\{ \left\vert d^t(u) - \frac{t}{m}d(u) \right\vert \geq \varepsilon \frac{t}{m}d(u) \right\} \leq 2\exp\left( -\frac{1}{3} \varepsilon^2 \frac{t}{m}d(u)\right).
    \]
\end{lemma}

\begin{proof}
    Consider the (independents) binary random variables $X_v = \mathbbm{1}\{ v \in N^t(u)\}$, for every $v \in N(u)$.
    Observe that $\Exp{} X_v = \frac{t (m-1)!}{m!} = \frac{t}{m}$, and that $d^t(u) = \sum_{v \in N(u)} X_v$.
    By linearity of expectation we have that $\Exp{}d^t(u) = \sum_{v \in N(u)} \Exp{}X_v = \frac{t}{m}d(u)$.
    By Chernoff inequality \cite[Corollary 2.3.4]{vershynin2018high} we obtain the lemma.
\end{proof}

Fix a parameter $\varepsilon \in (0,1)$, and consider the following partition of the neighborhood $N(u)$:
\begin{itemize}
    \item $\Good^+(u) := \{v \in N(u) \mid d(v) \geq (1+3\varepsilon)\gamma d(u)\}$;
    \item $\Good^-(u) := \{v \in N(u) \mid d(v) \leq (1-2\varepsilon)\gamma d(u)\}$;
    \item $\Bad(u) := \{v \in N(u) \mid (1-2\varepsilon) \gamma d(u) \leq d(v) \leq (1+3\varepsilon)\gamma d(u)\}$.
\end{itemize}
We now show that, as long as the degrees are large enough, the probability of a vertex $v \in \Good^+(u)$ to give a credit to $u$ tends to $1$, while for a vertex $v \in \Good^-(u)$ tends to $0$.
This is a desirable behavior of the algorithm. 
For the vertices in $\Bad(u)$, their behavior is uncertain, which is why they introduce a bias in the $\gamma$-degree estimate.
Fortunately, it is empirically clear that this bias is negligible on real graphs (see \Cref{fig:estimator}).

\begin{theorem}\label{thm:credit_expectation}
For every vertex $u \in V$, let $\Credit(u)$ be the credits collected by $u$ after a sequence of $m$ edge insertion.
For every $0 < \varepsilon \leq \frac{1}{3}$ we have 
\[
\Exp{} \Credit(u) = \gammadeg(u) \pm \Big( \vert \Bad(u) \vert + 4\sum_{v \in \Good(u)} e^{-\frac{1}{3}\varepsilon^2\min\{d(u), d(v)\}}\Big),
\]
where $\Good(u) = \Good^+(u)\cup\Good^-(u)$.
\end{theorem}

\begin{proof}
Fix any edge $(u,v) \in E$.
We start defining the binary random variable $X_v$ that is $1$ if $v$ gives a credit to $u$ when the edge $(u,v)$ is inserted (\ie\ at time $T$ uniformly distributed in $\{1,\dots, m\}$), and $0$ otherwise.

Now assume that $(u,v)$ comes at time $t \in [m]$.
Let us consider the event $\mathcal{E}_t := \{d^t(u) \in (1\pm\varepsilon)\tfrac{t}{m}d(u) \land d^t(v) \in (1\pm\varepsilon)\tfrac{t}{m}d(v)\}$.
By Lemma~\ref{lemma:concentration_degree} we have that $\Prob{}\{\mathcal{E}_t\} \geq 1 - 4\exp(-\frac{1}{3}\varepsilon^2\frac{t}{m}\min\{d(u),d(v)\})$.
If $\mathcal{E}_t$ holds, both the following conditions $d^t(v) \geq (1-\varepsilon)\frac{t}{m}d(v)$ and $d^t(u) \leq (1+\varepsilon)\frac{t}{m}d(u)$ are true.
It follows that $d^t(v) \geq \gamma d^t(u)$ is guaranteed if $d(v) \geq \frac{1+\varepsilon}{1-\varepsilon}\gamma d(u)$. At the same time, since $\varepsilon \leq \frac{1}{3}$, $d(v) \geq (1+3\varepsilon) \gamma d(u)$ implies $d(v) \geq \frac{1+\varepsilon}{1-\varepsilon}\gamma d(u)$\footnote{In fact, for $0 < \varepsilon \leq \frac{1}{3}$, we have that $1+3\varepsilon \geq \frac{1+\varepsilon}{1-\varepsilon}$.}.
In other words, if $v \in \Good^+(u)$ and the event $\mathcal{E}_t$ holds then $v$ gives the credit to $u$ at time $t$.
By marginalizing with respect to the insertion time of the edge $(u,v)$, we obtain
\begin{align*}
&\Prob{}\{X_v = 1 \mid v \in \Good^+(u)\}
= \frac{1}{m} \sum_{t=1}^{m}\Prob{}\{\mathcal{E}_T \mid T=t, v \in \Good^+(u)\}\\
&\geq \frac{1}{m} \sum_{t=1}^{m} 1-4 e^{\frac{1}{3}\varepsilon^2 \frac{t}{m}\min\{d(u),d(v)\}}
= 1 - \frac{4}{m}e^{-\frac{1}{3}\varepsilon^2\min\{d(u),d(v)\}}\sum_{t=1}^{m}e^{-t/m}\\
&\geq 1 - 4e^{-\frac{1}{3}\varepsilon^2\min\{d(u),d(v)\}},
\end{align*}
where the last inequality holds since $\sum_{i=1}^{m} e^{-t/m} < m$.
In the same way, if $v \in \Good^{-}(u)$ we have
\[
\Prob{}\{X_v = 1 \mid v \in \Good^{-}(u)\} \leq 4e^{-\frac{1}{3}\varepsilon^2\min\{d(u),d(v)\}}.
\]
We can now rewrite the expected number of credits of $u$ as
\[
\Exp{}\Credit(u) = \sum_{v \in N(u)} \Exp{}X_v = \sum_{v \in \Good^+(u)}\Exp{}X_v + \sum_{v \in \Good^-(u)}\Exp{}X_v + \sum_{v \in \Bad(u)}\Exp{}X_v.
\]
Thus, we have
\begin{align*}
&\sum_{v \in \Good^+(v)} \!\!\! 1-4e^{-\frac{1}{3}\varepsilon^2\min\{d(u),d(v)\}} \leq \Exp{}\Credit(u)
\\
&\leq \vert \Good^+(u) \vert +\!\!\! \sum_{v \in \Good^-(v)}\!\!\! 4e^{-\frac{1}{3}\varepsilon^2\min\{d(u),d(v)\}} + \vert \Bad(u) \vert.
\end{align*}
Since $\vert \Good^+(u) \vert \leq \gammadeg(u) \leq \vert \Good^+(u) \vert + \vert \Bad(u) \vert$, we finally obtain the theorem
\[
\vert \Exp{}\Credit(u) - \gammadeg(u) \vert \leq \vert \Bad(u) \vert + \sum_{v \in \Good(u)} 4e^{-\frac{1}{3}\varepsilon^2\min\{d(u),d(v)\}}. \qedhere
\]
\end{proof}

\subsection{Incremental algorithm for large quasi-cliques}

The \Cref{alg:init} initializes all the variables needed for the dynamic algorithm. More precisely, for every vertex $u \in V$ we have: the counter $\Credit(u)$ for the credits that $u$ holds, $\widetilde\Credit(u)$ that indicates the number of credit that $u$ had last time that $u$ was \emph{explored} to extract its quasi-clique, $\sigma(u)$ that is the $k$-min-hash signature of $N(u)$, and $S_u$ that is the quasi-clique associated to $u$.
Observe that for every vertex $u$ the space needed is $O(d(u) + k)$.
During the execution, we also maintain $S^*$, which is the largest quasi-clique found so far by the algorithm.

When a new edge $(u,v)$ is inserted, we first update the $k$-min-hash signatures $\sigma(u)$ and $\sigma(v)$, then the credits $\Credit(u)$ and $\Credit(v)$ as described before.
At this point we have to decide whether to \emph{explore} the vertices $u$ and $v$ by looking at all their neighborhoods $N(u)$ and $N(v)$ to update their quasi-cliques $S_u$ and $S_v$, as in \texttt{FastNBSim} algorithm.
Clearly, exploring the vertices at every insertion would result in an update time of $\Omega(k\Delta)$. To improve on that, our heuristics is based on the idea that we need to explore a vertex $u$ only if (i) the number of its credits $\Credit(u)$ has increased enough since the last time $u$ was explored (\ie\ $\Credit(u) \geq (1+\delta)\widetilde\Credit(u)$, for some parameter $\delta > 0$), and (ii) the number of credits $\Credit(u)$ is large enough with respect to the size of the current solution $\vert S^* \vert$ (\ie\ $\Credit(u) \geq \phi\vert S^* \vert$, for some parameter $\phi > 0$).
\Cref{alg:insert} implements the insert operation.

Point (i) ensures that we do not explore vertices too often, and \Cref{lemma:ammortized-cost-incremental} shows that this policy allows us to amortize the cost of the insertion operation.
Point (ii) instead relies on the fact that the $\gamma$-degree of a vertex $u$ is an upper-bound to the size of its quasi-clique $\vert S_u \vert$, so we want to avoid exploring vertices with $\gamma$-degrees that are too small.

\begin{lemma}\label{lemma:ammortized-cost-incremental}
    The amortized cost of the insert operation after any insertion sequence is $O(k\frac{\log{\Credit_{\max}}}{\delta})$, where $\Credit_{\max}$ is the maximum number of credits at the end of the sequence, and $k$ is the size of the $k$-min-hash signature.
\end{lemma}

\begin{proof}
    At each insertion, we update two $k$-min-hash signatures, and it takes $O(k)$ time.
    Throughout an entire insertion sequence of $m$ edges, we explore a single vertex $O(\log_{1+\delta} \Credit^m(u)) = O(\log_{1+\delta} \Credit_{\max})$ times, where  $\Credit_{\max} = \max_{u \in V} \Credit^m(u)$.
    By Taylor's expansion, for small values of $\alpha$, we can approximate $O(\log_{1+\delta}\Credit_{\max}) = O(\frac{\log\Credit_{\max}}{\delta})$.
    Since every time we explore a vertex $u$ we have an extra cost of $O(kd(u))$, we obtain an amortized cost of
    \begin{align*}
        &\frac{2mk + \sum_{u \in V}\frac{\log\Credit_{\max}}{\delta} \cdot O(kd(u))}{m}
        = 2k + \frac{\log\Credit_{\max}}{\delta}k\left( \frac{O\left(\sum_{u \in V} d(u)\right)}{m}\right)\\
        &=2k + \frac{\log\Credit_{\max}}{\delta}k \cdot O(1)
        = O\left( k\frac{\log\Credit_{\max}}{\delta} \right).
    \end{align*} 
\end{proof}

\begin{algorithm}
    \caption{\texttt{init}}
    \label{alg:init}
    \begin{flushleft}
    \textbf{Input:} an empty graph $G$.
    \end{flushleft}
    \begin{algorithmic}[1]
        \ForAll{$u \in V$}
            \State $\Credit(u) \gets 1$
            \State $\widetilde{\Credit}(u) \gets 1$
            \State $S_u \gets \{u\}$
            \State $\sigma(u) \gets \sigma(N(u))$
        \EndFor
        \State $S^* \gets \text{any } S_x$
    \end{algorithmic}
\end{algorithm}
\begin{algorithm}
    \caption{\texttt{insert}}
    \label{alg:insert}
    \begin{flushleft}
    \textbf{Input:} a new edge $(u,v)$, and two parameters $\delta,\phi \in (0,1)$.
    \end{flushleft}
    \begin{algorithmic}[1]
        \State $E \gets E \cup \{(u,v)\}$
        \State update $\sigma(u)$ and $\sigma(v)$
        \For{$x \in \{u,v\}$}
            \State $y \in \{u,v\} \setminus \{x\}$
            \IIf{$d(y) \geq \gamma d(x)$}
                $\Credit(x) \gets \Credit(x) + 1$
            \EndIIf
            \If{$\Credit(x) \geq (1+\delta)\widetilde{\Credit}(x) \land \Credit(x) \geq \phi \vert S^* \vert$}
                \State $\widetilde{\Credit}(x) \gets \Credit(x)$
                \State $S_x \gets \texttt{extract}(x, \gamma, b)$
                \IIf{$\vert S_x \vert > \vert S \vert$}
                    $S^* \gets S_x$
                \EndIIf
            \EndIf
        \EndFor
    \end{algorithmic}
\end{algorithm}

\begin{note}
The algorithm can be generalized and take as input a non-empty graph $G$.
The only necessary change is in the \texttt{init} procedure: just set $\Credit(u)$ and $\widetilde{\Credit}(u)$ equal to the actual $\gamma$-degree $\gammadeg(u)$, $S_u$ as the output of the procedure $\texttt{extract}(u,\gamma,b)$, and $S^*$ as the largest of the $S_u$.
\end{note}

\section{Extending the credit approach to the fully dynamic case}\label{sec:fully}

In this section we show how to adapt the Credit-Based approach of \Cref{sec:incremental} to support both edge insertions and edge deletions.

Adapting the incremental algorithm requires addressing two main challenges: (i) assigning and removing credits, so that a vertex's credits are still, in some way, a sort of index indicating how potentially profitable it is to explore that vertex (as happens with $\gammadeg$), and (ii) deciding how many times to explore a vertex to avoid performance degradation.

Regarding the first point, we adopt the simple strategy of decrementing $\Credit(u)$ by $1$ if, upon the removal of edge $(u,v)$, the vertex $v$ had previously given a credit to $u$.
Note that this condition can be verified in constant time if, for each neighbor $v$ of $u$, we store a bit indicating whether $v$ gave credit to $u$ when $(u,v)$ was inserted.
Although this idea might seem naive, it still empirically turns out to provide a good estimator for the $\gamma$-degrees (see \Cref{fig:estimator}).

The second challenge is more delicate.
As mentioned, an exploration policy that allows vertices to be explored as little as possible while simultaneously ensuring good solution quality would be ideal.
On the one hand, we would like a lazy strategy that guarantees some amortization.
At the same time, it is desirable to explore a vertex $u$ when its quasi-clique $S_u$ has potentially changed significantly since the last time $u$ was explored.
Formally, to prevent a vertex from being exported too many times, the algorithm maintains for each vertex $x$ a counter $\texttt{count}_x$ which is incremented each time the number of credits $\Credit(x)$ is modified (incremented or decremented).
At this point we impose that a vertex $x$ is free to be explored only when the condition $\widetilde{\Credit}(x) + \texttt{count}_x \geq (1+\delta)\widetilde{\Credit}(x)$ holds.
Note that, when the update sequence consists only of edge insertions, this condition is completely equivalent to the $\Credit(x) \geq  (1+\delta)\widetilde{\Credit}(x)$ condition of the incremental \Cref{alg:insert} at Line~6.
Now, we impose a second condition, namely $\max\{\Credit(x), \vert S_x \vert \} \geq \phi \vert S^* \vert$, which is designed to explore only those vertices $x$ for which the quasi-clique $S_x$ has potentially changed a lot, where $S^*$ is largest solution along all the $S_x$.
As before, if $\Credit(x) \geq \phi \vert S^* \vert$ means that $x$ has a high $\gamma$-degree, then $S_x$ could be potentially large.
On the other hand, if the $\gamma$-degree of $x$ has dropped significantly, then $S_x$ may no longer be a good (or valid) solution. This case is instead captured by $\vert S_x \vert \geq \phi \vert S^* \vert$. Intuitively, this condition captures the case where $S_x$ was a large solution, but given many operations on $x$, it may have been adversely affected, and therefore it is necessary to re-explore $x$ in order to guarantee a good quality of the solution.
Again, note that whenever $\max\{\Credit(x), \vert S_x \vert \} = \Credit(x)$ this condition coincides exactly with the condition in the incremental \Cref{alg:insert}.
This last equality is true on average, since $\Credit(x)$ is a good estimate of $\gammadeg(x)$, and because $\vert S_x \vert \leq \gammadeg(x)$.

Finally, when both conditions $\widetilde{\Credit}(x) + \texttt{count}_x \geq (1+\delta)\widetilde{\Credit}(x)$ and $\max\{\Credit(x), \vert S_x \vert \} $ $\geq \phi \vert S^* \vert$ are true, we can update $\widetilde{\Credit}(x)$, reset $\texttt{count}_x$, re-explore $x$ to extrapolate the new $S_x$, and update $S^*$.

\begin{algorithm}
    \caption{(new) \texttt{init}}
    \label{alg:init_fully_dynamic}
    \begin{flushleft}
    \textbf{Input:} an empty graph $G$.
    \end{flushleft}
    \begin{algorithmic}[1]
        \State initialize a priority queue \texttt{PQ}
        \ForAll{$u \in V$}
            \State $\Credit(u) \gets 1$
            \State $\widetilde{\Credit}(u) \gets 1$
            \State $\texttt{count}_u \gets 0$
            \State $S_u \gets \{u\}$
            \State $\texttt{PQ.set}(u, \vert S_u \vert)$
            \State $\sigma(u) \gets \sigma(N(u))$
        \EndFor
    \end{algorithmic}
\end{algorithm}
\begin{algorithm}
    \caption{\texttt{update}}
    \label{alg:update}
    \begin{flushleft}
    \textbf{Input:} an edge $(u,v)$, two parameters $\delta,\phi \in (0,1)$, and $\oplus \in \{-,+\}$.
    \end{flushleft}
    \begin{algorithmic}[1]
        \State update $E$ accordingly
        \State update $\sigma(u)$ and $\sigma(v)$
        \For{$x \in \{u,v\}$}
            \State $y \in \{u,v\} \setminus \{x\}$
            
            \State \textbf{if} $([\oplus = + \textbf{ and } d(y) \geq \gamma d(x)]$ \textbf{or}
            \Statex \hspace*{\algorithmicindent} \;\;\,\, $[\oplus = - \textbf{ and } y \text{ gave a credit to } x])$  \textbf{then}
                \State\hspace*{\algorithmicindent}\!\! $\Credit(x) \gets \Credit(x) \oplus 1$
                \State\hspace*{\algorithmicindent}\!\! $\texttt{count}_x \gets \texttt{count}_x + 1$
                
            \State \textbf{if} $(\widetilde{\Credit}(x) + \texttt{count}_x \geq (1+\delta)\widetilde{\Credit}(x)$ \textbf{and}
            \Statex \hspace*{\algorithmicindent} \;\;\,\,$\max\{\Credit(x), \vert S_x \vert \} \geq \phi \cdot  \texttt{PQ.getMax}())$ \textbf{then}
                \State\hspace*{\algorithmicindent}\!\! $\widetilde{\Credit}(x) \gets \Credit(x)$
                \State\hspace*{\algorithmicindent}\!\! $\texttt{count}_x \gets 0$
                \State\hspace*{\algorithmicindent}\!\! $S_x \gets \texttt{extract}(x, \gamma, b)$
                \State\hspace*{\algorithmicindent}\!\! $\texttt{PQ.updateKey}(x, \vert S_x \vert)$
        \EndFor
    \end{algorithmic}
\end{algorithm}

Note that in the fully dynamic case, there is no monotonicity in the size of $S^*$ over time, meaning it could grow and shrink. Therefore, we keep all $S_x$ sorted by size using a priority queue, and get the size of $S^*$ in constant time from the queue.
The complete pseudo-code of the update is reported in \Cref{alg:update}, while \Cref{alg:init_fully_dynamic} initializes all the new variables.
The \texttt{update} procedure includes both insertion and deletion, so it also takes as input $\oplus \in \{-,+\}$, that is, the type of operation to perform ($+$ for insertion and $-$ for deletion).

The final piece of the puzzle is updating the $k$-min-hash signatures.
While it's easy to update the $k$-min-hash signatures for incremental sequences, it's not possible to do so directly if there are also removals.
Jaccard similarity estimation in fully dynamic contexts has been widely studied, and various solutions have been proposed over the years \cite{BSS20,CGPS25,VOS,MROS,BOTBINDynSCAN}.
We use the recent \emph{Buffered min-hash} \cite{CGPS25}, currently the only data structure that allows for a fully dynamic min-hash, at the cost of minimal memory and little extra update time.

\section{Experimental results}\label{sec:experiments}
In this section we evaluate our approach on real datasets. We first assess the effectiveness of the incremental algorithm and then analyze the performance of our solution in the fully dynamic setting.

\paragraph{Platform.}
Our experiments were performed on a machine with 2.3 GHz Intel Xeon Gold 5118 CPU with 24 cores, 192 GB of RAM. The whole code is written in \texttt{C++}, compiled with \texttt{GCC 10} with optimization flag \texttt{-O3}.

\paragraph{Datasets}\label{par:dataset}

We evaluate our algorithms on both real and synthetic update sequences. The datasets are downloaded from SNAP \cite{snap}, NetworkRepository\footnote{\url{https://networkrepository.com}} and DyReach project \cite{HanauerHS20}.
For real data we consider two incremental datasets, \texttt{HPDyn} and \texttt{linux}, and one fully dynamic dataset, \texttt{wiki}.
Directed graphs are treated as undirected, and we remove self-loops and
multi-edges during pre-processing.
The main characteristics of the datasets are reported in
\Cref{tab:dataset-stats}.

To generate additional fully dynamic workloads, we start from static real graphs and construct synthetic update sequences using the following scheme.
We generate a random permutation of the edges and insert the first half of them. Then, at each step, we insert the next edge in the permutation with probability $1-p$, or delete uniformly at random one of the currently present edges with probability $p$. The final graph is therefore a subgraph of the original static graph. 
Other types of synthetic sequences are discussed and tested in \Cref{sec:further_exp}.

\begin{table*}
    \centering
    \caption{Dataset characteristics after pre-processing. For fully dynamic datasets, $\vert E \vert$ denotes the total number of updates (insertions and deletions).}
    \label{tab:dataset-stats}
    \begin{tabular}{llrrcc}
        \toprule
        \textbf{Dataset} & \textbf{Full Name} & \textbf{$\vert V \vert$} & \textbf{$\vert E\vert$} & \textbf{Type} & \textbf{Directed} \\
        \midrule
        FB & Ego-facebook & 4,039 & 88,234 & Static & \ding{56} \\ 
        HP & Ca-HepPh & 12,008 & 118,521 & Static &  \ding{56} \\  
        CM & Ca-CondMat & 23,133 & 93,497 & Static &  \ding{56}\\
        ER & Email-Enron & 36,692 & 183,831& Static &  \ding{56}\\
        GW & Loc-Gowalla & 196,591 & 950,327 & Static & \ding{56} \\
        SF & Web-Stanford & 281,903 & 2,312,497 & Static & \ding{52} \\
        LJ & LiveJournal & 3,997,692 & 34,681,189 & Static & \ding{52} \\
        \midrule
        wiki & Wikipedia Simple & 100,350 & 1,571,518 & Fully Dynamic& \ding{52}  \\
        linux & Linux & 100,000 & 159,996 & Incremental& \ding{52}  \\
        HPDyn & ca-cit HepPh & 100,000 & 3,148,487 & Incremental& \ding{56}  \\
        \bottomrule
    \end{tabular}
\end{table*}

\subsection{\texorpdfstring{$\gamma$}{gamma}-degree estimation}

The efficiency of our algorithms relies on the quality of the $\gamma$-degree estimator. In this section we empirically validate the estimator used in the incremental and fully dynamic algorithms (\Cref{sec:incremental,sec:fully}).

For each dataset, we execute the entire update sequence. For every vertex $u$, we compute the average number of accumulated credits over the $10$ independents runs, and compare it with the true $\gamma$-degree of $u$ in the final graph.

\Cref{fig:estimator} shows the results for the incremental dataset \texttt{linux} and the fully dynamic dataset \texttt{wiki}; results for the remaining datasets are reported in 
Appendix~\ref{apx:more_exp}.
Overall, the accumulated credits closely approximate the $\gamma$-degree and, most importantly, preserve the relative ordering of the vertices.

\begin{figure}[t!]
    \centering  
    \includegraphics[width=.8\linewidth]{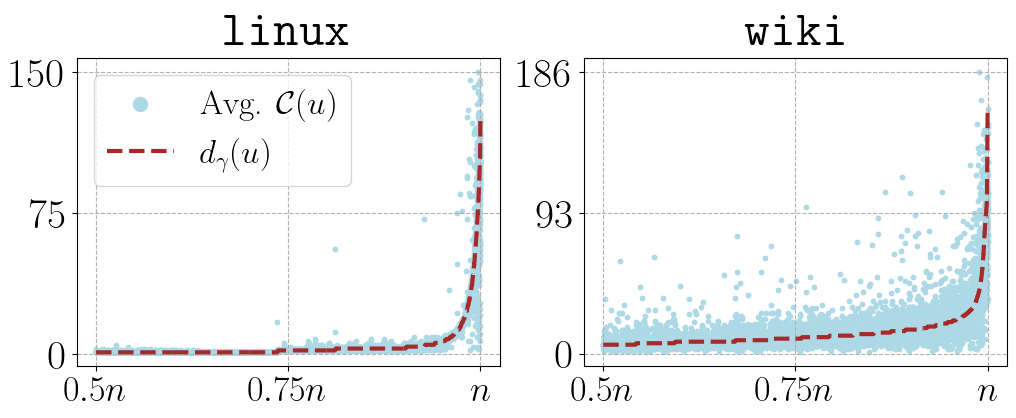}
    \caption{Average number of credits over $10$ runs. The $x$-axis reports vertices sorted by $\gamma$-degree, while the $y$-axis shows the average number of credits. Results are shown for \texttt{linux} (incremental) and \texttt{wiki} (fully dynamic). This representation highlights that the estimator preserves the relative ordering of vertices.}
    \label{fig:estimator}
    \Description{Average number of credits.}
\end{figure}

Some variance is observed, particularly in the incremental case (\texttt{linux}), where the estimator may slightly underestimate the $\gamma$-degree, while in the fully dynamic sequence (\texttt{wiki}) it may occasionally overestimate it.
Nevertheless, the estimator tracks well the $\gamma$-degree of high-$\gamma$ vertices, which are the most relevant for detecting large quasi-cliques~\cite{pang}.
    
\subsection{Parameter impact and setting}
The setting of the parameters $\delta, \phi, k$ governs the quality and performance of the algorithm.
The parameter $\phi$ controls the threshold on the accumulated credits used to trigger the extraction phase.
Since the $\gamma$-degree estimator may slightly underestimate the true value, choosing $\phi < 1$ prevents the algorithm from discarding promising vertices due to estimation noise.
This aspect becomes particularly relevant in the fully-dynamic setting, where fluctuations are more pronounced.
In practice, $\phi$ acts as a safety margin: smaller values make the algorithm more permissive, while values too close to $1$ may cause the algorithm to miss relevant quasi-cliques candidates.
However, decreasing $\phi$ excessively provides little benefits, as the relevant vertices are already explored for moderate values of the parameter.

The parameter $\delta$ mainly affects the running time.
It controls the amortization of the exploration costs by limiting how often vertex are explored.
As shown in \Cref{fig:tradeoffs}, the speed-up grows monotonically with $\delta$, since larger values reduce the number of explorations performed by the algorithm. 
While this monotonic behavior is theoretically guaranteed in the incremental setting (\Cref{lemma:ammortized-cost-incremental}), our experiments show that it also holds empirically on fully-dynamic sequences.

Finally, the parameter $k$ of the $k$-min-hash signatures influences both the quality and running time.
Small values of $k$ may lead to inaccurate containment estimates, especially on graphs with high-degree vertices, which in turn may degrade the quality of the reported quasi-cliques.
Conversely, large values of $k$ increase the computational cost without improving the estimation when neighborhoods are small. 
Hence, $k$ introduces an additional quality-performance trade-off.
The effect of this parameter has been studied in the static setting by Pang et al.~\cite{pang}. 

To study the quality-performance trade-offs, we performed an extensive set of experiments with several combinations of parameters. The performances of different parameter settings on three different datasets are summarized in \Cref{fig:tradeoffs}. 
Based on this analysis, we set $\delta = 0.3$, $\phi = 0.8$ and, $k = 64$, which provide the best overall trade-offs in our experiments.

Finally, we set $\gamma= 0.9$ and $b = 0.6$ for all the algorithm, following the configuration suggested in \cite{pang} .

\begin{figure}
    \centering
    \includegraphics[width=\linewidth]{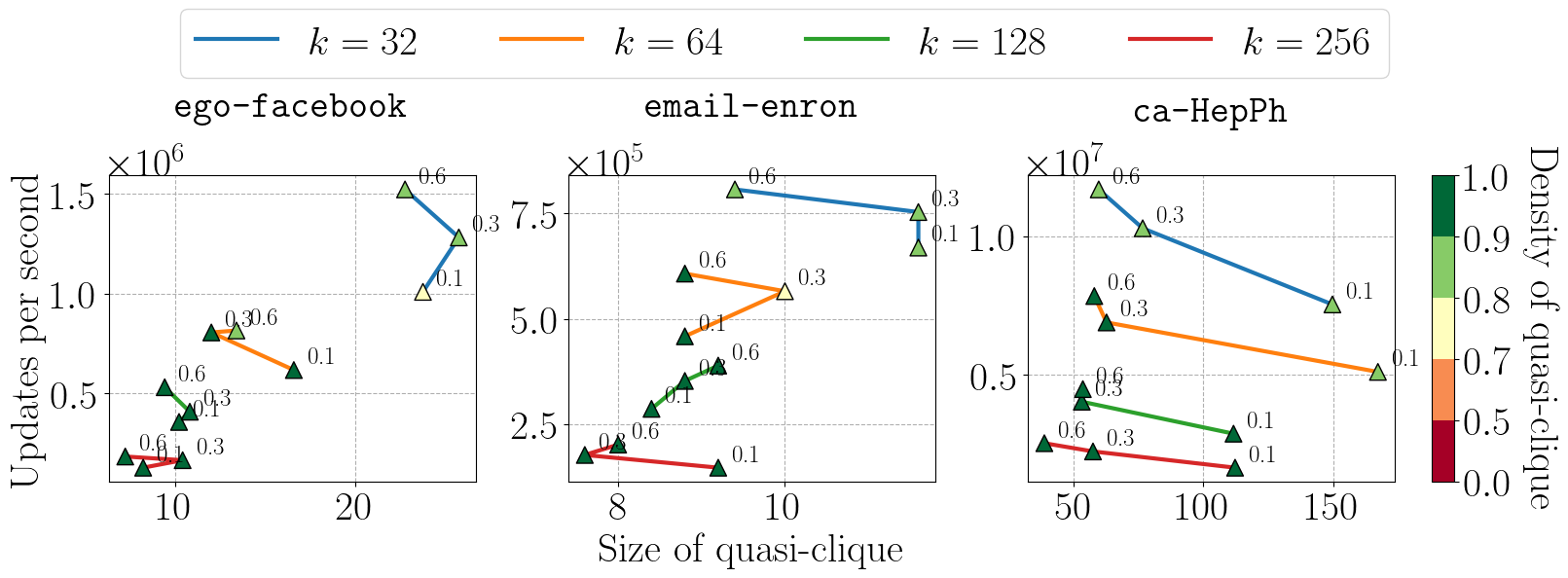}
    \caption{Different performance-quality tradeoffs obtained by the Credit-Based algorithm with different sets of parameters $k$ and $\delta$ on 3 datasets. On the x-axis the size of the returned quasi-clique, on the y-axis the number of updates per second. The triangles are density markers. Each curve corresponds to a different value of $k$. Each curve has three markers corresponding to the values of $\delta = 0.1, 0.3, 0.6$. For the first two datasets the sequence of updates used are two subgraph sequences with $p=0.1$ while the third one is an incremental dataset.}
    \label{fig:tradeoffs}
    \Description{Different performance-quality tradeoffs obtained by the Credit-Based algorithm with different sets of parameters $k$ and $\delta$ on 3 datasets.}
\end{figure}

\subsection{Comparison with the Baseline}\label{subsec:incremental_exp}
In this section we compare the two proposed dynamic algorithms, Baseline and Credit-Base, in order to study the relationship between the quality of the solution (\ie\ the size and density of the reported quasi-clique), and the performance.
In this section we consider \emph{only} real dynamic graphs.
A set of several synthetic sequences are deferred to \Cref{sec:further_exp}.

\begin{figure}
    \centering
    \includegraphics[width=\linewidth]{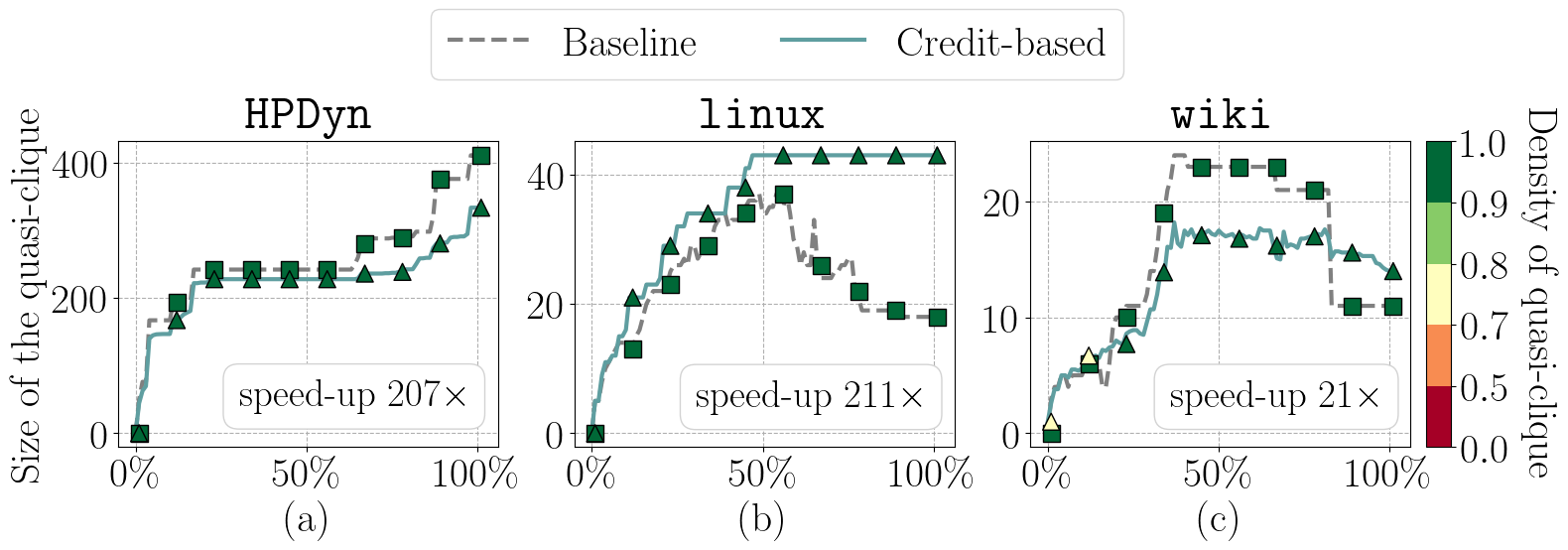}
    \caption{The plots illustrate the evolutionary size and density $\alpha$-quasi‑cliques found by the Credits-Based algorithm and the Baseline approach on three real dynamic sequences. The x‑axis shows the number of insertions/deletions processed; the y‑axis reports the size of the largest quasi‑clique found. The triangles are density markers and their color encodes subgraph density. The graphs (a) and (b) are incremental while (c) is fully dynamic.}
    \label{fig:quality}
    \Description{The plots illustrate the evolutionary size and density $\alpha$-quasi‑cliques found by the Credits-Based algorithm and the Baseline approach on three real dynamic sequences.}
\end{figure}

\Cref{fig:quality} shows representative results for the incremental dataset \texttt{linux} and \texttt{HPDyn}, and the fully dynamic dataset \texttt{wiki}.
In the incremental setting (\Cref{fig:quality} (a) and (b)), the Credit-Based algorithm achieves substantial speed-up over the Baseline while maintaining solutions of comparable quality.
In Particular, the density of the reported $\alpha$-quasi-cliques consistently remains in the high-density range (over $0.9$), showing that the algorithm does not sacrifice solution quality for efficiency.

The key idea is that the credit mechanism acts as a filter for small or irrelevant updates.
The condition requiring that the credits of a vertex to exceed $\phi \vert S^* \vert$, allows the algorithm to ignore edge insertions that do not significantly affect large quasi-cliques.
As a result, the algorithm avoids unnecessary recomputation and achieves speed-ups of up to $200\times$.
The trade-off is that, in some cases, the algorithm may miss slightly larger quasi-cliques; however with appropriate parameter settings the difference with the baseline remains negligible.

In the fully-dynamic case (\Cref{fig:quality} (c)), the same approach still provides significant speed-ups, although smaller than in the incremental case.
Edge deletions may invalidate the current quasi-clique, and thus  require more frequent re-validations of the neighborhoods, as discussed in \Cref{sec:fully}.
Consequently, the algorithm must explore the graph more often in order to return good-quality solutions, reducing the achievable acceleration.

Despite this effect, the Credit-Based algorithm still substantially outperforms the Baseline, while preserving the overall quality of the solutions.
The observed differences in size and density of the $\alpha$-quasi-cliques are similar to those in the incremental setting, confirming that the fully dynamic extension empirically preserves good quality-speed-up trade-offs.

\subsection{Other experiments on synthetic sequences}\label{sec:further_exp}

To complement the experimental evaluation, and further validate the effectiveness and efficiency of our algorithm, we perform additional experiments on several synthetic sequences, comparing our solution with the baseline.

We start describing how we generate the update sequences starting from real networks.

\paragraph{Graph Augmentation Sequence.}
We first insert all edges following a random permutation of the edge set. Subsequently, at each step, we insert a new edge with probability $1-p$ or delete a currently existing edge, chosen uniformly at random, with probability $p$. New edges are selected according to one of the following strategies (following~\cite{ZhaoGRB0025}): (i) \textsc{Random-Random} (\textsc{RR}), where a non-existing edge is chosen uniformly at random; and (ii) \textsc{Degree-Random} (\textsc{DR}), where the first endpoint $u$ is selected with probability $d(u)/2m$ and the second endpoint is chosen uniformly at random.

\paragraph{Results.}
We performed the same experiments as in \Cref{subsec:incremental_exp}, comparing the Credit-Based algorithm with the baseline.
We evaluate the size and density of the quasi-cliques reported by the algorithms over time, as well as the overall speed-up achieved by the Credit-Based algorithm relative to the baseline.
The parameter configuration for all experiments was set to $\delta = 0.3$, $\phi = 0.8$, and $k=64$.

The results are reported in \Cref{fig:further_exp_permutation,fig:further_exp_fully_dyn,fig:further_exp_degree_random,fig:further_exp_random_random}.
We observe behaviors similar to those described in \Cref{subsec:incremental_exp}. In sparse graph settings, such as at the beginning of the sequence, $\alpha$-quasi-cliques are relatively small. In these cases, our algorithm occasionally produces low-density $\alpha$-quasi-cliques. This occurs because, with small cliques, even a minor misestimation of edge counts significantly impacts the calculated density. This is fundamentally linked to the min-hash property, where approximation errors are more pronounced when the underlying set is small. However, we consider this scenario of limited interest, as naive algorithms typically perform better on such sparse graphs \cite{HanauerHS22}.

\begin{figure}
    \centering
    \includegraphics[width=\linewidth]{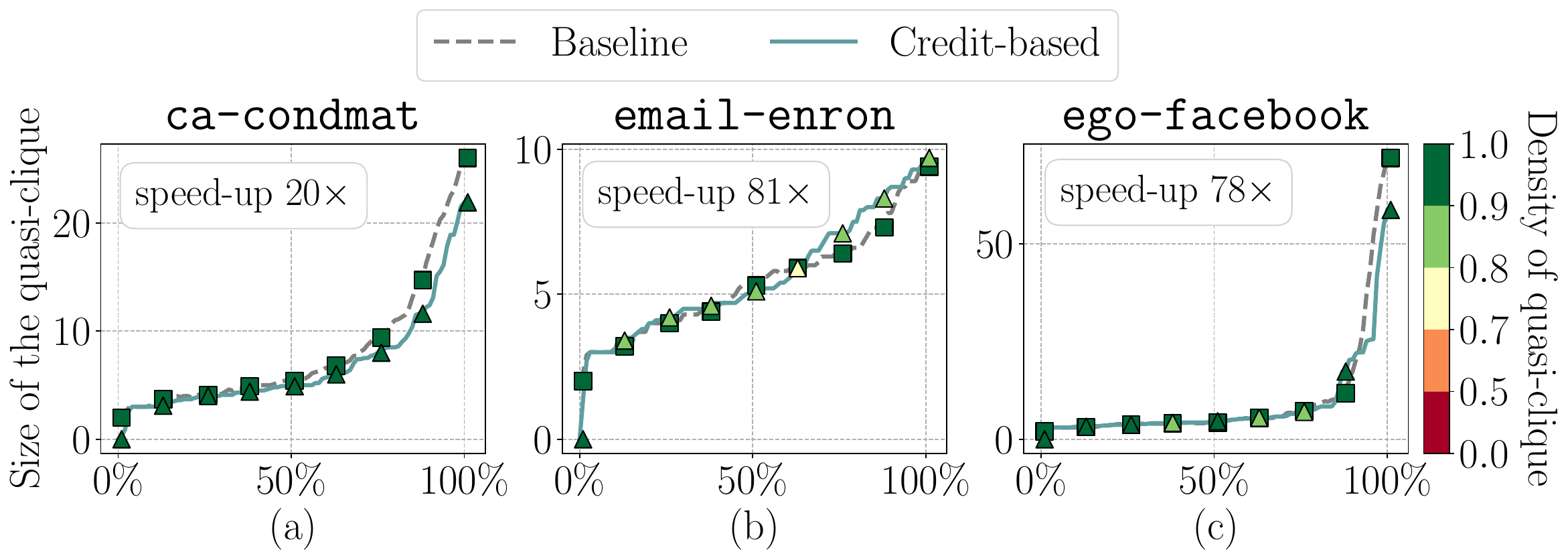}
    \caption{Performance comparison between the Credit-Based algorithm and the baseline using a sequence of edge insertions. The insertion order is determined by a random permutation of the edge set.}
    \label{fig:further_exp_permutation}
    \Description{Performance comparison between the Credit-Based algorithm and the baseline using a sequence of edge insertions.}
\end{figure}

\begin{figure}
    \centering
    \includegraphics[width=\linewidth]{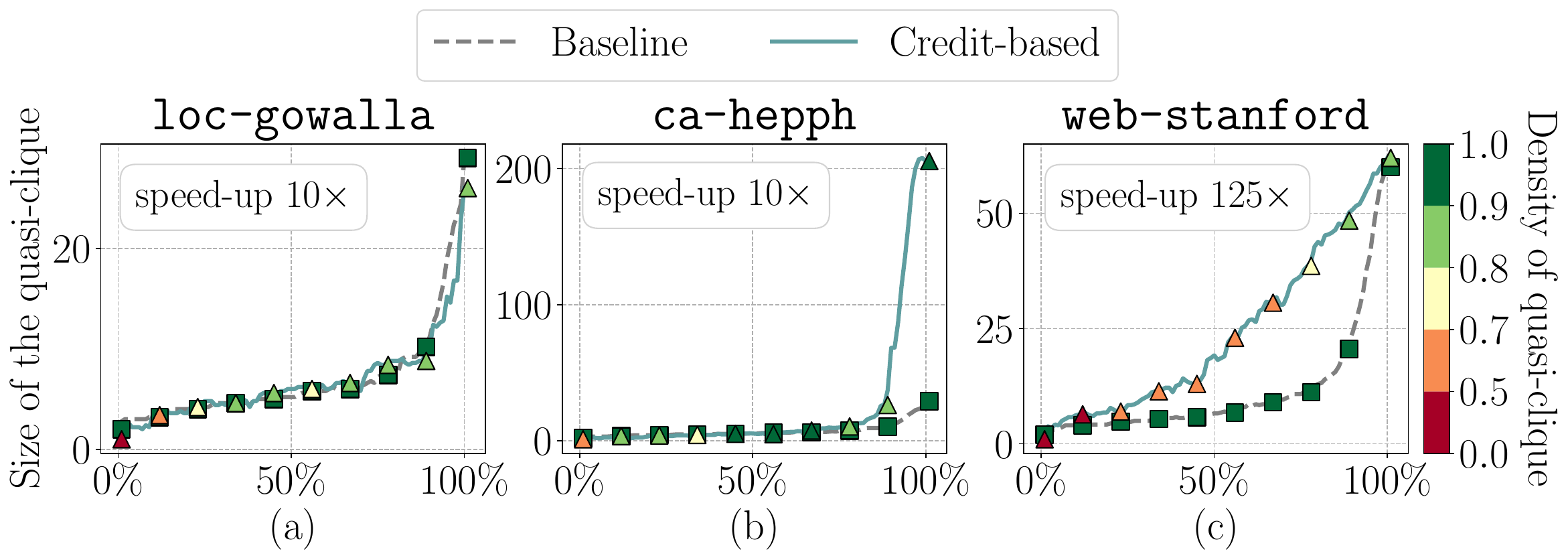}
    \caption{Comparison of the Credit-Based algorithm and the baseline using synthetic fully dynamic workloads. Update sequences are generated by randomly permuting edges: after an initial insertion of $50\%$ of the edges, subsequent updates consist of insertions with probability $1-p$ and deletions with probability $p=0.1$.}
    \label{fig:further_exp_fully_dyn}
    \Description{Comparison of the Credit-Based algorithm and the baseline using synthetic fully dynamic workloads.}
\end{figure}

\begin{figure}
    \centering
    \includegraphics[width=\linewidth]{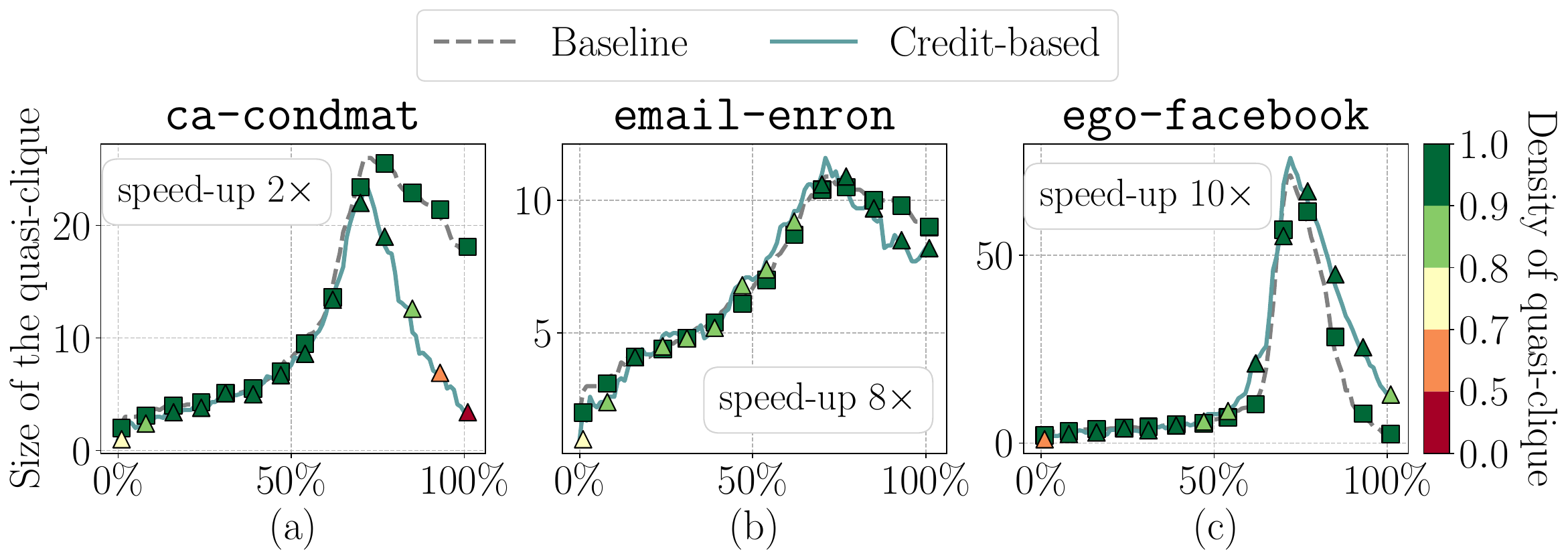}
    \caption{Performance comparison between the Credit-Based algorithm and the baseline using \textsc{Graph Augmentation} sequences. $20\%$ of new edges are added according to the \textsc{Degree-Random} distribution. The probability of deletion is $p = 0.1$.}
    \label{fig:further_exp_degree_random}
    \Description{Performance comparison between the Credit-Based algorithm and the baseline using \textsc{Graph Augmentation} sequences.}
\end{figure}

\begin{figure}
    \centering
    \includegraphics[width=\linewidth]{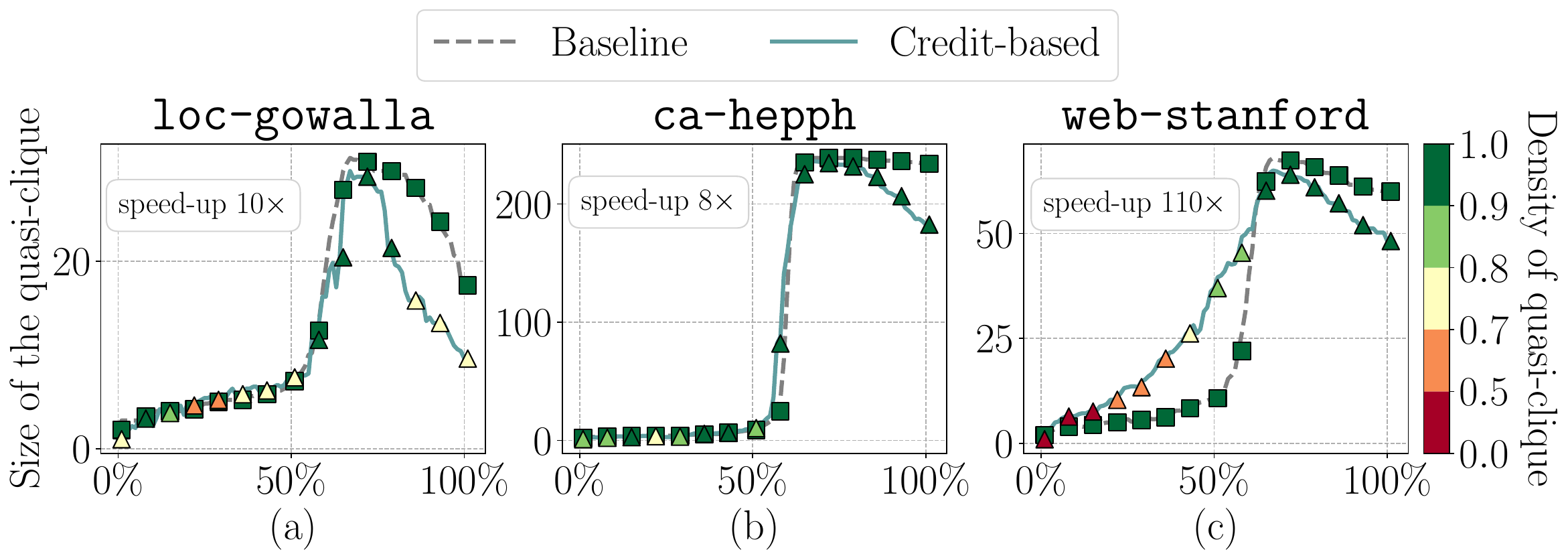}
    \caption{Performance comparison between the Credit-Based algorithm and the baseline using \textsc{Graph Augmentation} sequences. $20\%$ of new edges are added according to the \textsc{Random-Random} distribution. The probability of deletion is $p = 0.1$.}
    \label{fig:further_exp_random_random}
    \Description{Performance comparison between the Credit-Based algorithm and the baseline using \textsc{Graph Augmentation} sequences.}
\end{figure}


\clearpage

\balance
\bibliographystyle{ACM-Reference-Format}
\bibliography{ref}

\clearpage
\appendix
\section{Deferred Proofs}\label{apx:deferred_proofs}

\subsection{Proof of Lemma~\ref{lm:cs_update_delete}}

\begin{proof}
Let $N'(x)$ be neighborhood of $x$ after the deletion.
Since only the edge $(x,y)$ is removed, we have $N'(x) = N(x) \setminus \{y\}$, 
so $|N'(x)| = |N(x)| - 1$. For any $w \notin \{x,y\}$, the neighborhood 
$N'(w) = N(w)$ remains unchanged, hence $|N'(w)| = |N(w)|$.
We distinguish two cases based on whether $w \in N(y)$.

\paragraph{Case 1: $w \in N(x) \cap N(y)$.}
Since $w \in N(y)$, we have $y \in N(w) \cap N(x)$. After removing $(x,y)$, 
$y \notin N'(x)$, so $|N'(w) \cap N'(x)| = |N(w) \cap N(x)| - 1$. Therefore:
\[
    \tau'(x, w) = \frac{|N'(w) \cap N'(x)|}{|N'(x)|}  
         = \frac{|N(w) \cap N(x)| - 1}{|N(x)| - 1} 
         = \frac{\tau(x, w)\,|N(x)| - 1}{|N(x)| - 1},
\]
\[
    \tau'(w, x) = \frac{|N'(w) \cap N'(x)|}{|N'(w)|} 
         = \frac{|N(w) \cap N(x)| - 1}{|N(w)|} 
         = \tau(w, x) - \frac{1}{|N(w)|}.
\]

\paragraph{Case 2: $w \in N(x) \setminus N(y)$.}
Since $w \notin N(y)$, $y$ does not contribute to $N(w) \cap N(x)$, 
so $|N'(w) \cap N'(x)| = |N(w) \cap N(x)|$. Therefore:
\[
    \tau'(x, w) = \frac{|N'(w) \cap N'(x)|}{|N'(x)|}  
         = \frac{|N(w) \cap N(x)|}{|N(x)| - 1} 
         = \frac{\tau(x, w)\,|N(x)|}{|N(x)| - 1},
\]
\[
    \tau'(w, x) = \frac{|N'(w) \cap N'(x)|}{|N'(w)|} 
         = \frac{|N(w) \cap N(x)|}{|N(w)|} 
         = \tau(w, x).
\]
\end{proof}

\subsection{Proof of Lemma~\ref{lm:cs_update_insert}}

\begin{proof}
Let $N'(x)$ be neighborhood of $x$ after the insertion.
Since only the edge $(x,y)$ is added, we have $N'(x) = N(x) \cup \{y\}$, 
so $|N'(x)| = |N(x)| + 1$. For any $w \notin \{x,y\}$, the neighborhood 
$N'(w) = N(w)$ remains unchanged, hence $|N'(w)| = |N(w)|$.
We distinguish two cases based on whether $w \in N(y)$.

\paragraph{Case 1: $w \in N(x) \cap N(y)$.}
Since $w \in N(y)$, we have $y \in N(w) \cap N'(x)$. After adding $(x,y)$, 
$y \in N'(x)$, so $|N'(w) \cap N'(x)| = |N(w) \cap N(x)| + 1$. Therefore:
\[
    \tau'(x, w) = \frac{|N'(w) \cap N'(x)|}{|N'(x)|}  
         = \frac{|N(w) \cap N(x)| + 1}{|N(x)| + 1} 
         = \frac{\tau(x, w)\,|N(x)| + 1}{|N(x)| + 1},
\]
\[
    \tau'(w, x) = \frac{|N'(w) \cap N'(x)|}{|N'(w)|} 
         = \frac{|N(w) \cap N(x)| + 1}{|N(w)|} 
         = \tau(w, x) + \frac{1}{|N(w)|}.
\]

\paragraph{Case 2: $w \in N(x) \setminus N(y)$.}
Since $w \notin N(y)$, $y$ does not contribute to $N(w) \cap N(x)$, 
so $|N'(w) \cap N'(x)| = |N(w) \cap N(x)|$. Therefore:
\[
    \tau'(x, w) = \frac{|N'(w) \cap N'(x)|}{|N'(x)|}  
         = \frac{|N(w) \cap N(x)|}{|N(x)| + 1} 
         = \frac{\tau(x, w)\,|N(x)|}{|N(x)| + 1},
\]
\[
    \tau'(w, x) = \frac{|N'(w) \cap N'(x)|}{|N'(w)|} 
         = \frac{|N(w) \cap N(x)|}{|N(w)|} 
         = \tau(w, x).
\]
\end{proof}

\section{Further Experiments}\label{apx:more_exp}
\subsection{\texorpdfstring{$\gamma$}{gamma}-degree estimation}

\Cref{apx:fig:insertion_only_gamma_degree,apx:fig:fully_dynamic_gamma_degree} present the experimental results for the $\gamma$-degree estimation using the credit-based estimator. Specifically, \Cref{apx:fig:insertion_only_gamma_degree} illustrates the performance on insertion-only sequences, while \Cref{apx:fig:fully_dynamic_gamma_degree} focuses on fully dynamic sequences.
In both figures, the sequences are generated according to the procedure explained in \Cref{sec:experiments}.

\begin{figure}[h!]
    \centering
    \begin{subfigure}[b]{0.3\linewidth}
        \centering
        \includegraphics[width=\textwidth]{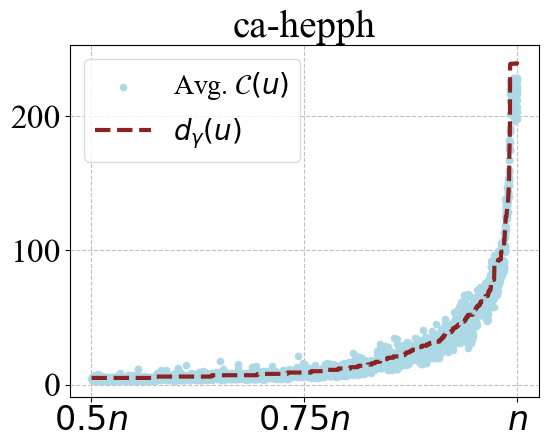}
        \label{apx:fig:condmat}
    \end{subfigure}
    \hfill
    \begin{subfigure}[b]{0.3\linewidth}
        \centering
        \includegraphics[width=\textwidth]{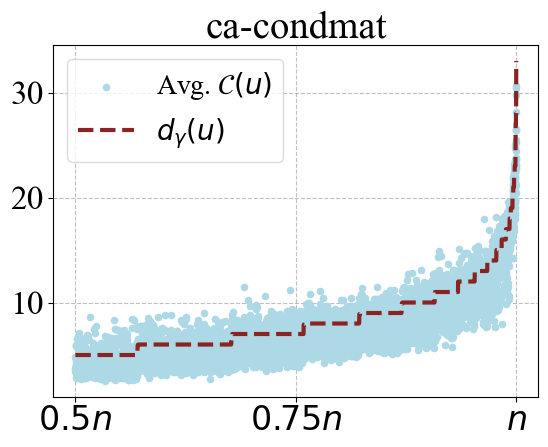}
        \label{apx:fig:hepph}
    \end{subfigure} 
    \hfill
    \begin{subfigure}[b]{0.3\linewidth}
        \centering
        \includegraphics[width=\textwidth]{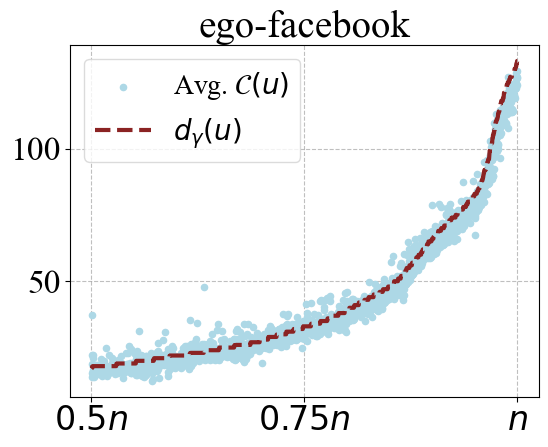}
        \label{apx:fig:facebook}
    \end{subfigure}
    \vspace{0.5cm} 
    \begin{subfigure}[b]{0.3\linewidth}
        \centering
        \includegraphics[width=\textwidth]{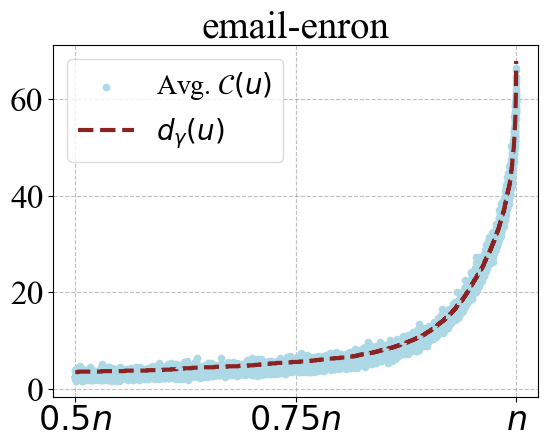}
        \label{apx:fig:enron}
    \end{subfigure}
    \hfill
    \begin{subfigure}[b]{0.3\linewidth}
        \centering
        \includegraphics[width=\textwidth]{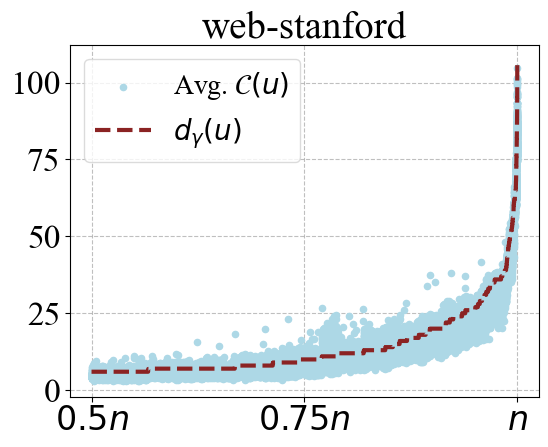}
        \label{apx:fig:linux}
    \end{subfigure}
    \hfill
    \begin{subfigure}[b]{0.3\linewidth}
        \centering
        \includegraphics[width=\textwidth]{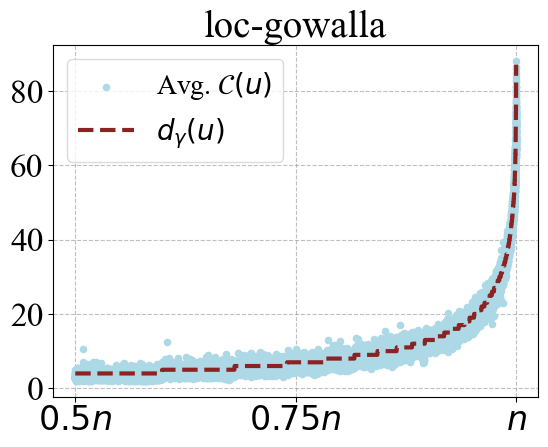}
        \label{apx:fig:gowalla}
    \end{subfigure}
    
    \caption{Average number of credits over $10$ runs. The $x$-axis reports vertices sorted by $\gamma$-degree, while the $y$-axis shows the average number of credits. Results are shown for static graphs. This representation highlights that the estimator preserves the relative ordering of vertices.}
    \label{apx:fig:insertion_only_gamma_degree}
    \Description{}
\end{figure}

\begin{figure}[h!]
    \centering
    \begin{subfigure}[b]{0.3\linewidth}
        \centering
        \includegraphics[width=\textwidth]{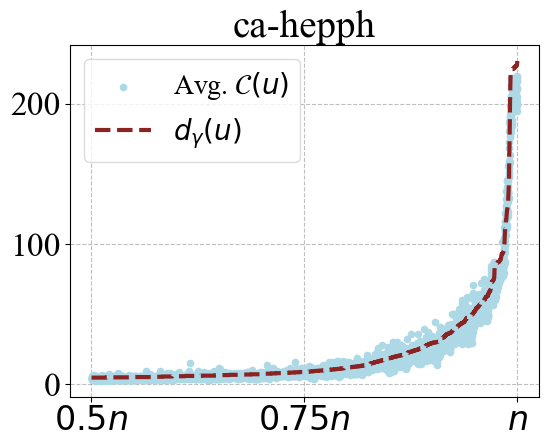}
        \label{fig:condmat}
    \end{subfigure}
    \hfill
    \begin{subfigure}[b]{0.3\linewidth}
        \centering
        \includegraphics[width=\textwidth]{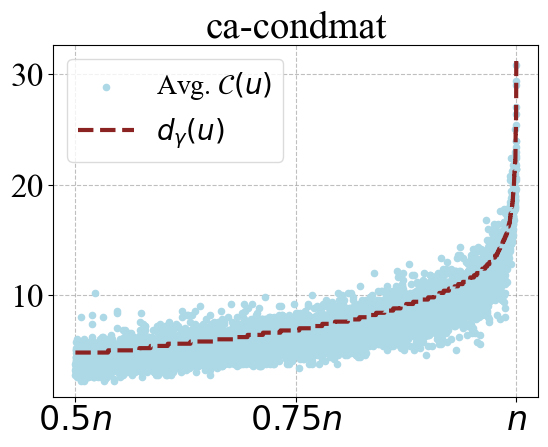}
        \label{fig:hepph}
    \end{subfigure} 
    \hfill
    \begin{subfigure}[b]{0.3\linewidth}
        \centering
        \includegraphics[width=\textwidth]{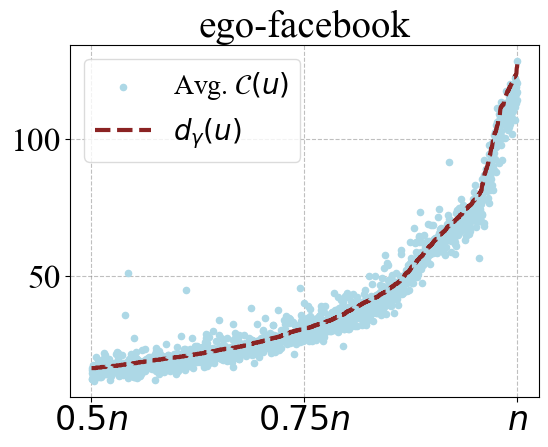}
        \label{fig:facebook}
    \end{subfigure}
    \vspace{0.5cm}
    \begin{subfigure}[b]{0.3\linewidth}
        \centering
        \includegraphics[width=\textwidth]{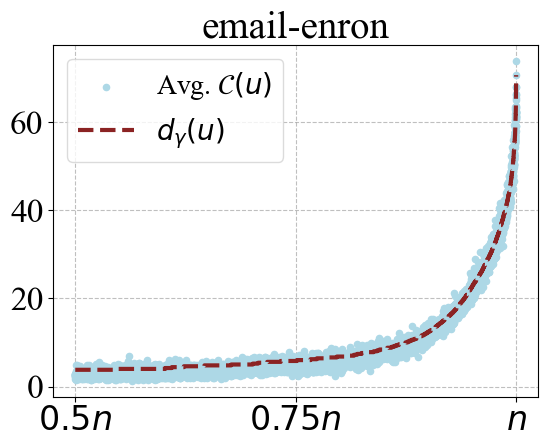}
        \label{fig:enron}
    \end{subfigure}
    \hfill
    \begin{subfigure}[b]{0.3\linewidth}
        \centering
        \includegraphics[width=\textwidth]{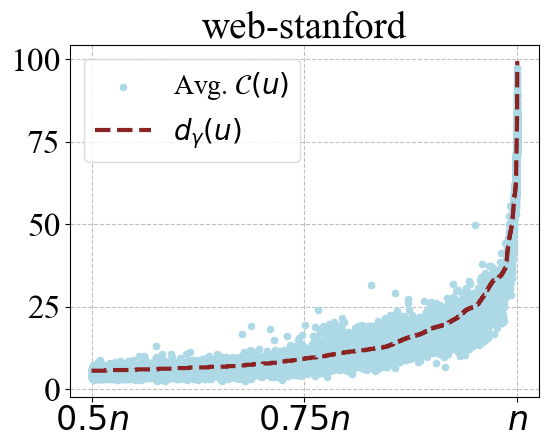}
        \label{fig:linux}
    \end{subfigure}
    \hfill
    \begin{subfigure}[b]{0.3\linewidth}
        \centering
        \includegraphics[width=\textwidth]{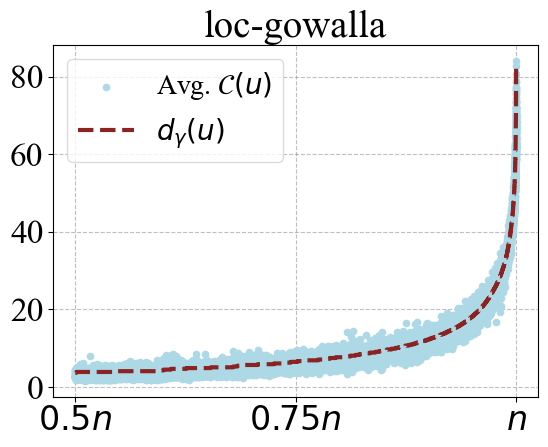}
        \label{fig:gowalla}
    \end{subfigure}
    
    \caption{Average number of credits over $10$ runs. The $x$-axis reports vertices sorted by $\gamma$-degree, while the $y$-axis shows the average number of credits. Results are shown for dynamic synthetics sequences with $p=0.1$. This representation highlights that the estimator preserves the relative ordering of vertices.}
    \label{apx:fig:fully_dynamic_gamma_degree}
    \Description{}
\end{figure}

\clearpage
\subsection{Comparison with baseline}
In this section, we extend the experimental evaluation on synthetic datasets from \Cref{sec:further_exp}, demonstrating the performance of our solution as the min-hash signature size $k$ varies.
Experimental results for \texttt{web-stanford} and \texttt{live-journal} are shown in \Cref{apx:fig:comparison_webstanford} and \Cref{apx:fig:comparison_livejournal}, respectively.

\begin{figure}
    \centering
    \includegraphics[width=.9\linewidth]{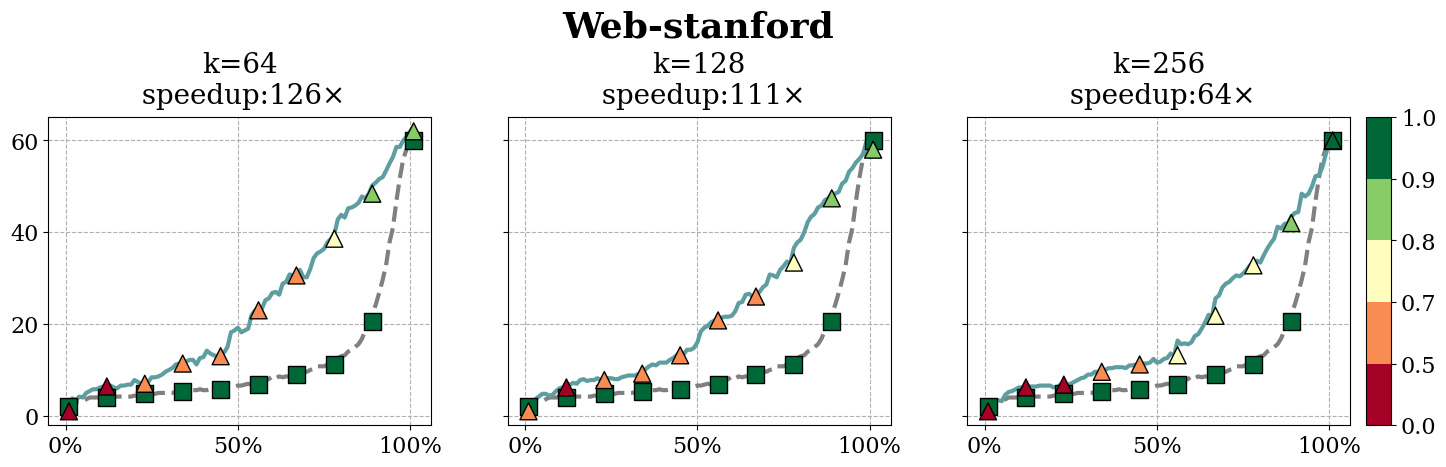}
    
    \caption{Quality of the solutions obtained using the random permutation sequence with $p=0.1$ using various value of $k$}
    \label{apx:fig:comparison_webstanford}
    \Description{}
\end{figure}

\begin{figure}
    \centering
    \includegraphics[width=.9\linewidth]{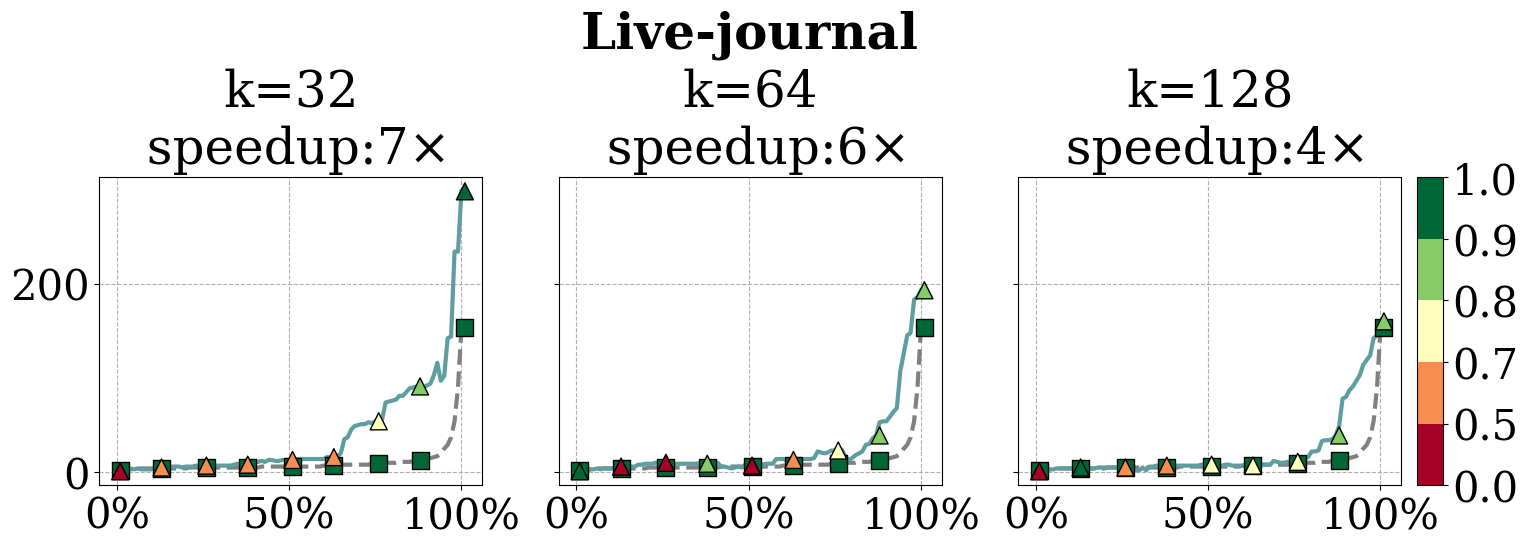}
    
    \caption{Quality of the solutions obtained using the incremental random permutation sequence using various value of $k$}

    \label{apx:fig:comparison_livejournal}
    \Description{}
\end{figure}

\end{document}